\documentclass[11pt]{article}
\usepackage[utf8]{inputenc}
\usepackage{setspace}
\usepackage[margin=1.25in]{geometry}
\usepackage{amsmath,amsfonts,amssymb,amsthm}
\usepackage[authoryear, longnamesfirst]{natbib}
\usepackage{hyperref}
\usepackage{graphicx}
\usepackage{enumitem}
\usepackage{nicematrix}
\usepackage{xcolor}
\usepackage{algorithm}
\usepackage{thmtools}
\usepackage{algorithmic}
\allowdisplaybreaks

\usepackage{mathtools}

\newtheorem{theorem}{Theorem}

\newtheorem{proposition}{Proposition}
\newtheorem{assumption}{Assumption}

\title{A Simple Approximation to the Distribution of the \\ \protect Ridge Regression Estimator\thanks{We would like to thank Bruce Hansen, Lihua Lei, J\"org Stoye, and Dacheng Xiu for very helpful comments and suggestions. Montiel Olea gratefully acknowledges financial support by the National Science Foundation Grant SES-2315600. }}
\author{Jos{\'e} Luis Montiel Olea \and Ryan Strong \and Amilcar Velez
\and Zhuoheng Xu \and Haomin Yu\thanks{All of the authors are affiliated to the Department of Economics at Cornell University. Corresponding author: {{\href{mailto:amilcare@cornell.edu}%
{amilcare@cornell.edu}}}}}
\date{\today}

\begin{document}

\maketitle

\onehalfspacing

\begin{abstract}
We present a simple Gaussian approximation to the finite-sample distribution of the classical ridge regression estimator. Our approximation captures the fact that, in finite samples, the ridge regression estimator trades off bias and variance to reduce estimation and prediction error. Our approximation is based on nonstandard asymptotics where $i)$ we let the estimator's regularization parameter grow proportionally to the sample size; and $ii)$ we treat the population regression coefficients as \emph{local} to the reference vector that defines the estimator's direction of shrinkage. In contrast to other asymptotic approximations in the literature, we allow for general forms of heteroskedasticity and autocorrelation in the data generating process (at the cost of considering a low-dimensional model where the number of covariates is not allowed to grow with the sample size). We use our simple Gaussian approximation to propose two new strategies to select the regularization parameter for the ridge regression estimator. The suggested strategies select the regularization parameter to minimize either average or worst-case excess prediction risk, where risk is computed using our suggested Gaussian approximation.  
\end{abstract}



\newpage 
\section{Introduction}

We have access to a dataset $D_n \equiv \{(y_i,x^{\top}_i)\}_{i=1}^{n}$ consisting of $n$ observations of a real-valued outcome variable, $y_i \in \mathbb{R}$, and a vector of $k$ covariates, $x_i \in \mathbb{R}^{k}$. The dataset $D_n$ is assumed to have been generated by the statistical model
\begin{equation} \label{eqn:statistical_model_introduction} 
y_i = x^{\top}_{i}\beta + \epsilon_i, \quad \{ (x^\top_i, \epsilon_i) \}_{i=1}^{n} \sim \mathbb{P}. 
\end{equation}
We are willing to restrict the data generating processes under consideration by requiring that 
\begin{equation} \label{eqn:assumptions_introduction} 
(1/n) \sum_{i=1}^{n} x_{i} x_{i}^{\top} \overset{p}{\rightarrow} \Sigma, \quad \textrm{ and } \quad  (1/\sqrt{n}) \sum_{i=1}^{n} x_{i} \epsilon_i \overset{d}{\rightarrow} \mathcal{N}_{k}(\mathbf{0},\Omega),
\end{equation}
where both $\Sigma$ and $\Omega$ are positive definite matrices, and $\mathcal{N}_{k}$ denotes a multivariate Gaussian distribution of dimension $k$. These assumptions allow for general forms of heteroskedasticity and autocorrelation in the data generating process.\footnote{We note that \eqref{eqn:statistical_model_introduction} and \eqref{eqn:assumptions_introduction} neither require nor imply that the conditional expectation of the outcome given the covariates is linear. In fact, \eqref{eqn:statistical_model_introduction} and \eqref{eqn:assumptions_introduction} can be viewed as the typical high-level assumptions used in the analysis of the \emph{linear projection model} using the terminology in \citet[Chapter~2, p.~37]{hansen2022econometrics}.} 

We are interested in approximating the finite-sample distribution of the \emph{ridge regression estimator:}
\begin{equation} \label{eqn: ridge estimator formula introduction}
\widehat{\beta}_{\lambda_n} \equiv \left( \sum_{i=1}^{n} x_i x_i^{\top} + \lambda_n \mathbb{I}_k \right)^{-1} \left( \sum_{i=1}^{n} x_i y_i + \lambda_n \beta_0\right). 
\end{equation}
We refer to the nonnegative scalar $\lambda_n$ as the \emph{regularization parameter} and to $\beta_0 \in \mathbb{R}^{k}$ as the \emph{reference vector}. We index the regularization parameter by the sample size to allow for the possibility that it depends on the dataset $D_n$.  

Our first result (Theorem \ref{theorem: big lambda_n approximation}) shows that there is a sense in which, under \eqref{eqn:statistical_model_introduction}-\eqref{eqn:assumptions_introduction}, we can approximate the distribution of $\widehat{\beta}_{\lambda_n}-\beta$ by the distribution: 
\begin{equation} \label{eqn: asy distribution ridge introduction }
 \mathcal{N}_{k} \left( -(\lambda_n/n) (\Sigma + (\lambda_n/n) \mathbb{I}_k)^{-1} (\beta-\beta_0) , (1/n)(\Sigma + (\lambda_n/n) \mathbb{I}_k)^{-1} \Omega (\Sigma + (\lambda_n/n) \mathbb{I}_k)^{-1} \right).
\end{equation}
 Our simple approximation captures the fact that, in finite samples, the choice of regularization parameter and reference vector affect the bias and variance of the ridge regression estimator. We view the mean and variance in the approximation in \eqref{eqn: asy distribution ridge introduction } as providing simple generalizations of common formulae for the bias and variance of the ridge estimator derived conditional on covariates and under i.i.d. sampling; see, for example, \citet[Equations 29.8 and 29.9]{hansen2022econometrics}. The Gaussian approximation in \eqref{eqn: asy distribution ridge introduction } can also be viewed as a generalization of the asymptotic distribution of the ridge estimator presented in the seminal work of  \citet{knight2000asymptotics}, which they derive under conditional homoskedasticity (in our framework, heteroskedasticity and serial autocorrelation are captured in the matrix $\Omega$). 

In order to derive Theorem \ref{theorem: big lambda_n approximation}, we use \emph{nonstandard asymptotics}.\footnote{See \cite{powell2017identification} for an overview of the role of nonstandard asymptotics in modern econometrics.} As we will explain later, we assume that the true coefficient $\beta$ in the model \eqref{eqn:statistical_model_introduction} (henceforth denoted $\beta_n$) depends on the sample size. We treat $\beta_n$ as  \emph{local to the reference vector $\beta_0$} that defines the estimator's direction of shrinkage. More concretely, we assume that:
\begin{equation} \label{eqn: local-to-beta0} 
\sqrt{n}(\beta_n-\beta_0) \rightarrow b,
\end{equation}
for some $b \in \mathbb{R}^{k}$. In addition to \eqref{eqn: local-to-beta0}, our asymptotic analysis accommodates data-dependent regularization parameters that grow proportionally to the sample size and that satisfy
\begin{equation} \label{eqn: ridge regularization limit} 
\lambda_n/n \overset{p}{\rightarrow } \lambda \in [0, \infty).  
\end{equation} 

Our second result (Theorem \ref{theorem: approximation to the prediction risk of ridge}) uses the approximation in \eqref{eqn: asy distribution ridge introduction } to provide an approximation to the \emph{excess prediction risk} of the ridge regression estimator. As we explain later, excess prediction risk is defined as the additional prediction risk relative to an oracle that knows $\beta_n$. Predictions of the outcome variable based on the ridge regression estimator take the form
\begin{equation} \label{eqn: predictions ridge}
\widehat{a}_{\lambda_n}(x) = x^{\top} \widehat{\beta}_{\lambda_n},
\end{equation}
where $\widehat{\beta}_{\lambda_n}$ is defined in \eqref{eqn: ridge estimator formula introduction}. 

The exact finite-sample analysis of the prediction risk of ridge regression is challenging even under a stylized homoskedastic, Gaussian regression model. Several recent papers in the statistics literature provide approximations and lower and upper bound for this prediction risk, including \cite{dobriban2018high}, \cite{hastie2022surprises}, \cite{mourtada2022elementary}, \cite{atanasov2024risk}. We show that it is possible to use our Theorem \ref{theorem: big lambda_n approximation} to provide an approximation to the excess prediction risk of ridge regression as a function of three types of parameters: $\lambda$ (the probability limit of $\lambda_n/n$); $b$ (the parameter controlling the local-to-$\beta_0$ regression coefficient $\beta_n$); and the variance parameters $\Sigma$ and $\Omega$. Importantly, the variance parameters can be consistently estimated, but the local parameter $b$ cannot.   

Our third result (Theorem \ref{theorem: selection of lambda}) presents closed-form, data-driven recommendations for the selection of the regularization parameter $\lambda_n$. Our formulae are designed for the special case in which---under the true unknown data generating process---the covariates are asymptotically uncorrelated and have the same variance; that is, $\Sigma = \sigma^2_x \mathbb{I}_k$, where $\mathbb{I}_k$ is the identity matrix of dimension $k$ and $\sigma^2_x$ is a common variance parameter. Following \cite{hastie2022surprises}, we refer to this setting as one with \emph{isotropic features}.\footnote{Although we note that in \cite{hastie2022surprises} $\Sigma$ equals the identity matrix, whereas we only require $\Sigma = \sigma^2_{x} \mathbb{I}_k$.} The key step in deriving our closed-form formulae is to choose the regularization parameter to optimize the excess prediction risk approximation in Theorem \ref{theorem: approximation to the prediction risk of ridge}. As we discussed before, this approximation depends on the local parameter $b$ and variance parameters $\sigma^2_x$ and $\Omega$. While the latter parameters can typically be consistently estimated under mild assumptions, the local parameter $b$ cannot. We use the common decision-theoretic principles of average and worst-case risk, following \citet{ferguson_1967}, to account for the unknown parameter $b$ in the approximate excess prediction risk. If $\pi$ is a prior distribution on $b$ such that $\mathbb E_\pi[b^\top b]<\infty$ and $\mathbb E_\pi[b^\top b]>0$, Theorem \ref{theorem: selection of lambda} shows that the $\pi$-optimal selection of the regularization parameter $\lambda_n$ in ridge regression is
\begin{equation} \label{eqn: isotropic-pi-optimal introduction}
\widehat{\lambda}_{n,\pi}^* =n \cdot \frac{\operatorname{trace}(\widehat{\Omega})}{\widehat{\sigma}_x^2 \mathbb E_\pi[b^\top b]}.
\end{equation}
When the localization parameter $b$ admits a known bound $B$ on its norm ($\|b\| \leq B$), the selection of the regularization parameter $\lambda_n$ that minimizes approximate worst-case excess risk (and that we term the \emph{minimax} choice of $\lambda_n$) is:
\begin{equation} \label{eqn: isotropic-minimax-optimal introduction}
\widehat{\lambda}_{n, \mathrm{minimax}}^*=n \cdot \frac{\operatorname{trace}(\widehat{\Omega})}{\widehat{\sigma}_x^2 B^2}.
\end{equation} 

While the formulae in \eqref{eqn: isotropic-pi-optimal introduction}-\eqref{eqn: isotropic-minimax-optimal introduction} pertain to the case of isotropic features, the idea of optimizing approximate excess prediction risk (either by using an average or worst-case criterion) is more general. We explain how to operationalize both of these approaches for nonisotropic features; see Sections \ref{subsec:minimax} and \ref{subsec:average}. In both cases, the optimal choice of regularization parameter is the solution of a nonconvex optimization problem over the positive part of the real line. In this case, we suggest choosing $\lambda_n$ by evaluating the objective functions on a grid of candidate regularization parameters. This approach is the same as the one used by conventional statistical packages to find and implement the usual cross-validated choice of regularization parameter.  

{\scshape Related Literature:} 
The introduction of the ridge estimator to regression analysis is often credited to \cite{hoerl1970ridge}, although the estimator is based on the earlier work of \cite{horel1962application}; see \cite{hoerl2020ridge}. The textbook version of the ridge regression estimator uses a reference vector $\beta_0 = 0_{k \times 1}$; see, for example, \citet[Equation~3.41, p.63]{ESL2017} or \citet[Chapter 29.5]{hansen2022econometrics}. The more general formulation given in \eqref{eqn: ridge estimator formula introduction} has several precedents in the literature; for example, \citet[Equation 2.1]{swindel1976good}. See also \cite{anup1984ridge}.  

The high-level assumptions in \eqref{eqn:assumptions_introduction} that we use to analyze the distribution of the ridge estimator are common in the econometrics literature.  For example, it is well known that \eqref{eqn:assumptions_introduction} can be verified under i.i.d. sampling and standard primitive conditions on the joint distribution of $x_i$ and $\epsilon_i$; see \citet[Assumption~7.2 and Theorems~7.1 and~7.2]{hansen2022econometrics}. However, our framework accommodates richer data structures, allowing departures from i.i.d. sampling as well as general forms of heteroskedasticity and autocorrelation in $(x^{\top}_i,\epsilon_i)$.

The approximation we propose in \eqref{eqn: asy distribution ridge introduction } has not---to the best of our knowledge---appeared before in the literature expressed at this level of generality. The closest reference that we are aware of is the seminal work of \cite{knight2000asymptotics}, who present a Gaussian approximation to the distribution of the ridge regression estimator assuming independence between $x_i$ and $\epsilon_i$ (and independence across observations). The distribution in \eqref{eqn: asy distribution ridge introduction } coincides with their result when specialized to the case in which $\Omega$ is conditionally homoskedastic (that is, $\Omega = \mathbb{E}[\epsilon_i^2] \Sigma$) and $\beta_0=0_{k \times 1}$. We view the mean and variance in the approximation in \eqref{eqn: asy distribution ridge introduction } as providing simple generalizations of common formulae for the bias and variance of the ridge estimator derived conditional on covariates and under i.i.d. sampling; see, for example, \citet[Equations~29.8 and~29.9]{hansen2022econometrics} and also \citet[Section~4]{hoerl1970ridge}. 

The idea of considering drifting sequences of parameter values---as those in \eqref{eqn: local-to-beta0} and \eqref{eqn: ridge regularization limit}---in order to improve the distributional approximations provided by standard asymptotic theory has a long history in econometrics. See \cite{powell2017identification} for examples and details on the use of nonstandard asymptotics in econometrics. More concretely, the idea of using nonstandard asymptotics to analyze the distribution of the ridge regression estimator also appears in \cite{knight2000asymptotics} (they refer to it as \emph{triangular array asymptotics}), and more recently in the work of \cite{shen2025weak}, who assess the predictive performance of several machine learning methods in high-dimensional regressions with low signal-to-noise ratios. Drifting sequences of parameter values are also common in the analysis of shrinkage estimators; see, for example, \citet[Equation~14,~p.~119]{hansen2016efficient}. Local asymptotics are also used in the work of \cite{hirano2017forecasting} to compare the performance of different forecasting procedures under model uncertainty.

There is also a recent literature analyzing the risk of predictions based on the ridge regression estimator, where the outcome variable predictions take the same form as in \eqref{eqn: predictions ridge}. For example, \cite{hsu2012random}, \cite{dobriban2018high}, \cite{hastie2022surprises}, \cite{mourtada2022elementary}, \cite{atanasov2024risk}. All these papers make distributional assumptions that preclude data generating processes that exhibit at least one of the following features: conditional heteroskedasticity, time-series autocorrelation, or a nonlinear conditional expectation function. It is important to mention, however, that the additional assumptions in these papers allow them to consider \emph{high-dimensional} approximations where the number of covariates can be large relative to the sample size. In contrast, all the analysis in our paper pertains to a \emph{low-dimensional} model where we treat the number of covariates as fixed, and let the sample size diverge to infinity.

Finally, the idea of using an approximation of the risk function to choose regularization parameters has a long history in statistics and econometrics; see, for example, the discussion in \citet[p.~744]{abadie2019choosing}. To the best of our knowledge the formulae we provide in \eqref{eqn: isotropic-pi-optimal introduction} and \eqref{eqn: isotropic-minimax-optimal introduction} are new to the literature. It is important to mention, however, that one of the initial motivations of our paper was to recover an expression analogous to the optimal regularization parameter for ridge regression given by \cite{hastie2022surprises} in their Corollary 6 to Theorem 6. Their results show that using high-dimensional asymptotics where $k/n \rightarrow \gamma$, the oracle choice of regularization parameter for the ridge estimator (assuming $x_i$ and $\epsilon_i$ are independent and that $\Sigma = \mathbb{I}_k$) takes the form $\lambda^* = n \mathbb{E}[\epsilon_i^2]\gamma/\| \beta \|^2$. Assuming conditional homoskedasticity, our formulae in \eqref{eqn: isotropic-pi-optimal introduction} and \eqref{eqn: isotropic-minimax-optimal introduction} become analogous to \cite{hastie2022surprises}'s formula, but the true unknown $\|\beta_n\|^2$ is replaced by either the prior mean of $\| b \|^2$ or the maximum value of $\|b\|^2$. More generally, we note that while ridge regression and lasso-type estimators with cross-validated tuning parameters are routinely used with data that exhibit heteroskedasticity and autocorrelation, it is not entirely clear what are the theoretical guarantees of cross-validation in these environments; see, for example, the discussion in \cite{kock2026data}.

{\scshape Outline:} The rest of this paper is organized as follows. Section \ref{sec:notation} introduces the notation, framework, and main assumptions. Section \ref{sec:main results} presents our main results. Section \ref{sec:simulations} presents simulation evidence to illustrate and support our main results. Section \ref{sec:conclusion} concludes.

\section{Notation and Framework} \label{sec:notation} 

An econometrician has access to a dataset $D_n \equiv \{(y_i,x^{\top}_i)\}_{i=1}^{n}$ consisting of $n$ observations of a real-valued outcome variable, $y_i \in \mathbb{R}$, and a vector of $k$ covariates, $x_i \in \mathbb{R}^{k}$. The dataset $D_n$ is assumed to have been generated by the statistical model
\begin{equation} \label{eqn:statistical_model} 
y_i = x^{\top}_{i}\beta_n + \epsilon_i, \quad \{ (x^\top_i, \epsilon_i) \}_{i=1}^{n} \sim \mathbb{P}_n. 
\end{equation}
The parameters of this statistical model are i) the unknown vector of slope coefficients, $\beta_n \in \mathbb{R}^{k}$, and ii) the unknown distribution of the tuple $((x^{\top}_1,\epsilon_1)$, \ldots, $(x^{\top}_n,\epsilon_n))$, which we denote as $\mathbb{P}_n$. Note that we have indexed both $\beta_n$ and $\mathbb{P}_n$ by the sample size $n$. While this does not make a difference when conducting finite-sample analysis (since, in that case, $n$ is fixed), it will allow us for more flexibility when considering asymptotic approximations to the distribution of different statistics. 

We implicitly restrict the parameter space of the statistical model in \eqref{eqn:statistical_model}  by requiring $\mathbb{P}_n$ to satisfy the following high-level assumption: 

\begin{assumption} [High-level assumptions on $\mathbb{P}_n$] \label{asn: assumption high-level OLS ridge}
The distribution $\mathbb{P}_n$ satisfies
\begin{enumerate}    
\item  $(1/n) \sum_{i=1}^{n} x_{i} x_{i}^{\top} \overset{p}{\rightarrow} \Sigma$ for a positive definite matrix $\Sigma$. 

\item $(1/\sqrt{n}) \sum_{i=1}^{n} x_{i} \epsilon_i \overset{d}{\rightarrow} \mathcal{N}_{k}(\mathbf{0},\Omega)$ for a positive definite matrix $\Omega$.

\end{enumerate}
\end{assumption}
The high-level Assumption \ref{asn: assumption high-level OLS ridge} is a convenient way to consider a large class of data generating processes in our analysis. It is well known that Assumption \ref{asn: assumption high-level OLS ridge} can be verified under i.i.d. sampling and standard primitive conditions on the joint distribution of $x_i$ and $\epsilon_i$; see \citet[Assumption 7.2 and Theorems 7.1 and 7.2]{hansen2022econometrics}. However, our framework accommodates richer data structures, allowing departures from i.i.d. sampling as well as general forms of heteroskedasticity and autocorrelation in $(x^{\top}_i,\epsilon_i)$. 

The continuous mapping theorem and Slutsky's theorem immediately imply the asymptotic normality of the least-squares estimator of $\beta_n$ in the model \eqref{eqn:statistical_model}; namely 
\begin{equation} \label{eqn: ols distribution}
\sqrt{n} \left( \widehat{\beta}_{\textrm{OLS}} - \beta_n \right) \overset{d}{\rightarrow} \mathcal{N}_{k} \left( \mathbf{0}, \Sigma^{-1} \Omega \Sigma^{-1} \right), \quad \widehat{\beta}_{\textrm{OLS}} \equiv \left( \sum_{i=1}^{n} x_i x_i^{\top} \right)^{-1} \sum_{i=1}^{n} x_i y_i.  
\end{equation}
This slightly generalizes the normal approximation to the distribution of the least-squares estimator derived under i.i.d. sampling; see, for example, \citet[Theorem~7.3]{hansen2022econometrics}. In our framework, $\Omega$ can be the long-run variance of the process $\{x_i\epsilon_i\}_{i=1}^{\infty}$, and $\beta_n$ can vary with the sample size.  

\section{Main Results} \label{sec:main results}

\subsection{Asymptotic Distribution of the Ridge Estimator}

Define the \emph{ridge estimator}---with regularization parameter $\lambda \in \mathbb{R}_{+}$ and a reference vector $\beta_0 \in \mathbb{R}^{k}$---to be the solution of the following minimization problem:
\begin{equation} \label{eqn: ridge minimization} 
\min_{\beta \in \mathbb{R}^k} \frac{1}{n}\sum_{i=1}^n (y_i-x_i^{\top}\beta)^2 + \frac{\lambda}{n}\sum_{j=1}^k (\beta_j-\beta_{0j})^2.
\end{equation}
The solution to the minimization problem in \eqref{eqn: ridge minimization} can be shown to equal
\begin{equation} \label{eqn: ridge estimator formula}
\widehat{\beta}_{\lambda} \equiv \left( \sum_{i=1}^{n} x_i x_i^{\top} + \lambda \mathbb{I}_k \right)^{-1} \left( \sum_{i=1}^{n} x_i y_i + \lambda \beta_0\right).
\end{equation}
It is known that the objective function that defines the ridge estimator intentionally sacrifices \emph{training error} (since predictors of $y_i$ based on $\widehat{\beta}_{\lambda}$ will have a larger training error than predictions based on the least-squares estimator $\widehat{\beta}_{\textrm{OLS}}$) by penalizing deviations away from the reference vector $\beta_0$. 

We note that in most textbook definitions of the ridge estimator, the reference vector $\beta_0$ equals a vector of zeros of dimension $k \times 1$; see, for example, \citet[Equation~3.41, p.63]{ESL2017} or \citet[Chapter 29.5]{hansen2022econometrics}. Since the usual interpretation of \eqref{eqn: ridge estimator formula} is that the ridge estimator ``shrinks'' the least-squares coefficients, we allow for the possibility that this shrinkage is towards an arbitrary reference vector $\beta_0$ that could be different from zero. We note that this more general formulation has several precedents in the literature; for example, \citet[Equation (2.1)]{swindel1976good}. See also \cite{anup1984ridge}.

As we mentioned in the introduction, the main goal of this paper is to present a useful approximation to the finite-sample distribution of the ridge estimator in \eqref{eqn: ridge estimator formula}. We are particularly interested in deriving an approximation that captures the well-known fact that, in finite samples, the ridge estimator presents a bias-variance trade-off. Textbook expressions of the bias and variance of the ridge estimator in finite samples are usually derived assuming i.i.d. sampling and conditional mean independence; i.e., $\mathbb{E}_{\mathbb{P}_n}[\epsilon_i | x_i] = \mathbb{E}_{\mathbb{P}_n}[\epsilon_i]=0$. See, for example, \citet[Chapter 29.6]{hansen2022econometrics}. This stands in contrast with the generality under which the asymptotic distribution of the least-squares estimator can be derived; see Equation \eqref{eqn: ols distribution} above.  

Our first result shows that when the regularization parameter $\lambda$ is \emph{negligible relative to the sample size} (even if it is data dependent), the standard asymptotic distribution of the ridge estimator will not be useful to capture the bias-variance trade-off that is present in finite samples. More precisely, the following result shows that the asymptotic distribution of the ridge estimator will be asymptotically equivalent to that of the least-squares estimator. In order to state our result (and to allow for the possibility that the regularization parameter $\lambda$ could be selected in a data-driven manner), we consider a possibly stochastic sequence of regularization parameters $\{\lambda_n\}_{n=1}^{\infty}$.  

\begin{proposition}[Approximation to the distribution of $\widehat{\beta}_{\lambda_n}$ when $\lambda_n/n$ is negligible]  \label{proposition: small lambda_n approximation}  
Suppose that the dataset $D_n$ was generated according to the statistical model \eqref{eqn:statistical_model}, and suppose that $\mathbb{P}_n$  satisfies Assumption \ref{asn: assumption high-level OLS ridge}. If 
\[ \lambda_{n}/n \overset{p}{\rightarrow}0 \: \textrm{ and } \: (\lambda_n/n) \sqrt{n}(\beta_n-\beta_0) \overset{p}{\rightarrow}0,  \]
then
\[ \sqrt{n} \left( \widehat{\beta}_{\lambda_n} - \beta_n \right) -  \sqrt{n} \left( \widehat{\beta}_{\textrm{OLS}} - \beta_n \right) \overset{p}{\rightarrow} 0.\]
\end{proposition}
\begin{proof}
See Appendix \ref{subsec: small lambda_n approximation}.
\end{proof} 
The intuition behind this result follows directly from the expression for $\widehat{\beta}_{\lambda_n}$ in \eqref{eqn: ridge estimator formula}. When $n$ is large and $\lambda_n/n$ is close to zero, then $\widehat{\beta}_{\lambda_n}$ is approximately the same as $\widehat{\beta}_{\textrm{OLS}}$. This is why the distribution of $\sqrt{n}(\widehat{\beta}_{\lambda_n} - \beta_n)$ is approximately the same as the distribution of the least-squares estimator. 

The proof of Proposition \ref{proposition: small lambda_n approximation} shows that if $\lambda_n$ is not negligible relative to the sample size, then the asymptotic distribution of the ridge estimator will differ from that of the least-squares estimator. Moreover, the asymptotic distribution will capture the bias-variance trade-off that the ridge estimator faces in finite samples. 

\begin{theorem}[Approximation to the distribution of $\widehat{\beta}_{\lambda_n}$ when $\lambda_n/n$ is potentially non-negligible] \label{theorem: big lambda_n approximation}
Suppose that the dataset $D_n$ was generated according to the statistical model \eqref{eqn:statistical_model}, and suppose that $\mathbb{P}_n$  satisfies Assumption \ref{asn: assumption high-level OLS ridge}. If 
\begin{equation} \label{eqn: key condition Theorem 1}
\lambda_{n}/n \overset{p}{\rightarrow} \lambda \in [0,\infty) \: \textrm{ and } \: \sqrt{n}(\beta_n-\beta_0) \rightarrow b \in \mathbb{R}^{k},  
\end{equation}
then
\begin{equation} \label{eqn: asy distribution ridge }
\sqrt{n} \left( \widehat{\beta}_{\lambda_n} -\beta_n \right) \overset{d}{\rightarrow} \mathcal{N}_{k} \left( -\lambda (\Sigma + \lambda \mathbb{I}_k)^{-1} b , (\Sigma + \lambda \mathbb{I}_k)^{-1} \Omega (\Sigma + \lambda \mathbb{I}_k)^{-1} \right) .
\end{equation}
    
\end{theorem}
    
\begin{proof}
See Appendix \ref{subsec: big lambda_n approximation}.
\end{proof} 

We view the mean and variance in the limiting distribution in \eqref{eqn: asy distribution ridge } as providing simple generalizations of common formulae for the bias and variance of the ridge estimator derived conditional on covariates and under i.i.d. sampling; see, for example, \citet[Equations (29.8) and (29.9)]{hansen2022econometrics}. 

The key assumptions of Theorem \ref{theorem: big lambda_n approximation} are the two conditions in \eqref{eqn: key condition Theorem 1}. The first condition requires the regularization parameter to be potentially non-negligible relative to the sample size. Note that we allow for the possibility that the regularization parameter is selected in a data-driven way, as long as it has a deterministic probability limit. We note that the idea of using sequences of regularization parameters $\lambda_n$ to analyze the asymptotic properties of penalized estimators is not ours and has several precedents in the statistics literature; see, for example, \citet[Theorems 2 and 3]{knight2000asymptotics} where different rate conditions on $\lambda_n$ are used to analyze the ridge estimator and more general lasso-type estimators.  

The second condition in \eqref{eqn: key condition Theorem 1} can be interpreted as saying that the true regression coefficient $\beta_n$ is ``local-to-$\beta_0$''. This means that we are assuming that the reference vector used to compute the ridge estimator in \eqref{eqn: ridge estimator formula} is not too far from the true (and unknown) $\beta_n$. We note that if $\sqrt{n}(\beta_n-\beta_0)$ is unbounded, then any non-negligible regularization parameter $\lambda_n$ leads to a bias in $\sqrt{n} \left( \widehat{\beta}_{\lambda_n} -\beta_n \right)$ that becomes arbitrarily large.

The idea of considering drifting sequences of parameter values in order to improve the distributional approximations provided by standard asymptotic theory has a long history in econometrics. See, for example, the local-to-zero asymptotics in the linear instrumental variables model of \cite{staiger1997instrumental}; the local-to-unit-root asymptotics in the study of nearly integrated autoregressive processes in \cite{phillips1988regression} and its recent generalization in \cite{dou2021generalized}; the local-to-identification-failure analysis in the study of nonlinear Generalized Method of Moments models with weak identification in \cite{andrews2022optimal}; and the growing-number-of-folds asymptotic framework in \cite{velez2024asymptotic}, which improves distributional approximations for Debiased Machine Learning estimators. See \cite{powell2017identification} for more examples and details on the use of nonstandard asymptotics in econometrics. Some recent work more directly related to our set-up is the paper of \cite{shen2025weak}, studying a high-dimensional linear regression model with \emph{weak} signals. While our work focuses on regression models where the dimension of the covariates is fixed, a version of the weak signals model in \cite{shen2025weak} can be obtained in the case in which the reference vector equals zero and the true parameter $\beta_n$ drifts towards zero at rate $1/\sqrt{n}$.

Although Theorem \ref{theorem: big lambda_n approximation} follows from elementary asymptotic theory, to our knowledge, we are not aware of any previous work that formally derives the simple approximation in \eqref{eqn: asy distribution ridge } at the level of generality allowed by Assumption \ref{asn: assumption high-level OLS ridge}. Textbook results for the bias and variance of the ridge estimator (usually derived conditional on covariates and under i.i.d. sampling of a linear regression model) can yield a Gaussian distribution similar to \eqref{eqn: asy distribution ridge } in finite samples if the error term in the regression model is also assumed to be Gaussian (conditional on the vector of covariates). But the analysis is more nuanced if the error distribution is non-Gaussian and/or there is serial autocorrelation in the regression residuals.

The closest result to \eqref{eqn: asy distribution ridge } that we have found in the literature appears in the  seminal paper by \cite{knight2000asymptotics}, which provides asymptotic results for lasso-type estimators (including the ridge estimator). On p. 1368 they use their Theorem 4 to derive an approximation to the asymptotic distribution of a ridge estimator with reference vector $\beta_0 = 0_{k \times 1}$. Under i.i.d. sampling and a conditional homoskedasticity assumption (i.e., $\Omega = \mathbb{E}[x_ix_i^{\top} \epsilon_i^2] = \mathbb{E}[x_ix_i^{\top}] \mathbb{E}[\epsilon_i^2]$), they obtain the same formula as in \eqref{eqn: asy distribution ridge } but specialized to the case in which $\Omega = \mathbb{E}[\epsilon_i^2] \Sigma $. Thus, our results can be viewed as a generalization of their analysis of the asymptotic distribution of the ridge estimator, but allowing for more general data structures (where heteroskedasticity and serial autocorrelation are permitted and captured by the matrix $\Omega$), and where we also allow for a general reference vector $\beta_0$. 
\cite{knight2000asymptotics} also consider an asymptotic sequence where $\beta_n$ changes with the sample size (they refer to this as \emph{triangular array asymptotics}). We note that a potentially interesting extension of our results is to derive asymptotic approximations for lasso-type estimators (analogous to part b) of their Theorem 4) under our Assumption \ref{asn: assumption high-level OLS ridge}. These results could then be used to provide concrete recommendations for the choice of regularization parameter.  

In the remaining part of the paper, we maintain our focus on the ridge estimator and use the asymptotic approximation in Theorem \ref{theorem: big lambda_n approximation} to present approximations to the \emph{prediction risk} of the ridge estimator, and we use statistical decision theory to provide recommendations on the choice of regularization parameter.  

\subsection{An Approximation to the Prediction Risk of Ridge Regression}  

In this section we use Theorem \ref{theorem: big lambda_n approximation} to provide an approximation to the \emph{prediction risk} of ridge regression. We start by providing a definition of the prediction problem we are interested in, and we then present a formal definition of the risk of predictions based on ridge regression. 

\emph{Prediction Problem.} An econometrician has access to a dataset $D_n \equiv \{(y_i,x^{\top}_i)\}_{i=1}^{n}$ consisting of $n$ observations of a real-valued outcome variable, $y_i \in \mathbb{R}$, and a vector of $k$ covariates, $x_i \in \mathbb{R}^{k}$. The dataset $D_n$ is assumed to have been generated by the statistical model \eqref{eqn:statistical_model} with parameters $(\beta_n, \mathbb{P}_n)$ that satisfy Assumption \ref{asn: assumption high-level OLS ridge}. A prediction function is a mapping $a:\mathbb{R}^{k} \rightarrow \mathbb{R}$ that maps covariates into a predicted value for the outcome variable. We note that a prediction function is defined for all possible values of covariates, including those that have not been observed in the sample. In order to define a loss function for the prediction problem, we introduce an additional assumption. 
\begin{assumption}\label{asn: stationary-prediction-law}
The stochastic process $\{(x^{\top}_i,\epsilon_i)\}_{i=1}^{\infty}$ is strictly stationary in the sense of Definition 1.3.3 in \cite{Brockwell_Davis:2013}. 
\end{assumption} 

The strict stationarity assumption means that---associated to $\mathbb{P}_n$---there exists a probability distribution $\mathbb P$ over $\mathbb{R}^{k+1}$ such that, for every $i=1,\dots,n,$ we have $(x_i^\top,\epsilon_i) \overset{\mathbb{P}_n}{\sim} \mathbb P$. We will refer to $\mathbb{P}$ as the \emph{stationary distribution} of $\mathbb{P}_n$. We use this stationary distribution to define a loss function for the prediction problem as follows. Suppose we sample a new pair $(x^{\top},\epsilon)$ according to the stationary distribution $\mathbb{P}$, and use the slope coefficient $\beta_n$ to construct a new outcome-covariate pair $(y,x)=(x^{\top}\beta_n+\epsilon,x )$.\footnote{Note then that in our analysis we do not allow the distribution used to evaluate prediction error to deviate arbitrarily from the distribution that generated the data. As shown by \cite{patil2024optimal}, such an assumption can have important implications over the optimal choice of regularization parameter. A framework for analyzing the distributionally robust prediction error for the square-root lasso and related estimators (and for optimally selecting tuning parameters) has been proposed in \cite{montiel2026distributionally}, but their framework excludes the ridge regression estimator.} We can then define the loss associated to a prediction function $a:\mathbb{R}^{k} \rightarrow \mathbb{R}$ as 
\[ L(a,\beta_n,\mathbb{P}) \equiv \mathbb{E}_{(y,x) \sim (\beta_n,\mathbb{P})} [ (y-a(x))^2 ].\]
In a slight abuse of notation, $(y,x) \sim (\beta_n,\mathbb{P})$ is used to capture the fact that the expectation is taken over outcome-covariate pairs generated using $(\beta_n, \mathbb{P})$ as described above. 

\emph{Prediction Risk.} In a prediction problem, the goal is to construct a data-driven prediction function; that is, we want to use the available training data to choose a prediction function. In a slight abuse of notation, we denote data-driven prediction functions as $\widehat{a}$. Formally, we think of $\widehat{a}$ as a mapping that takes the data set $D_n$ as input and returns a prediction function $\widehat{a}$ from covariates to outcomes. Because $\widehat{a}$ is constructed from the training data, it is a random function. The prediction risk is then the expected prediction loss 
\begin{equation} \label{eqn: prediction risk general}
R(\hat{a}; \beta_n, \mathbb{P}_n) = \mathbb{E}_{(\beta_n,\mathbb{P}_n)}[L (\widehat{a};\beta_n,\mathbb{P})],
\end{equation}
where $\mathbb{E}_{(\beta_n,\mathbb{P}_n)}$ means that the expectation is taken over the distribution of $D_n$, which is parameterized by $(\beta_n,\mathbb{P}_n)$.

\emph{Approximate Prediction Risk.} Consider the prediction function based on the ridge regression estimator
\begin{equation}
\widehat{a}_{\lambda_n}(x) = x^{\top} \widehat{\beta}_{\lambda_n},
\end{equation}
where $\widehat{\beta}_{\lambda_n}$ is defined in \eqref{eqn: ridge estimator formula} with $\lambda$ replaced by $\lambda_n$. 

The exact finite-sample analysis of the prediction risk of ridge regression is challenging even under the stylized homoskedastic, Gaussian regression model. However, there are different papers that provide approximations (and lower/upper bounds) to this prediction risk. For example, \cite{dobriban2018high}, \cite{hastie2022surprises}, \cite{mourtada2022elementary}, \cite{atanasov2024risk}. The following result shows that we can use Theorem \ref{theorem: big lambda_n approximation} to provide an approximation to the excess prediction risk of ridge regression as a function of three types of parameters: $\lambda$ (the probability limit of $\lambda_n/n$); $b$ (the parameter controlling the local-to-$\beta_0$ regression coefficient $\beta_n$); and the variance parameters $\Sigma$ and $\Omega$ defined in Assumption \ref{asn: assumption high-level OLS ridge}.  
More precisely, define the \emph{approximate excess risk function} 
\begin{equation}\label{eqn: approximation to the prediction risk of ridge}
R^{e}_n( \lambda; b,\Sigma,\Omega ) \equiv \frac{1}{n}\lambda^2 b^{\top}  (\Sigma + \lambda \mathbb{I}_k)^{-1} \Sigma  (\Sigma + \lambda \mathbb{I}_k)^{-1} b + \frac{1}{n} \textrm{trace} \left(  (\Sigma + \lambda \mathbb{I}_k)^{-1} \Omega  (\Sigma + \lambda \mathbb{I}_k)^{-1} \Sigma \right). 
\end{equation}

\begin{theorem}[Approximation to the excess prediction risk of ridge regression] \label{theorem: approximation to the prediction risk of ridge}
Suppose that the dataset $D_n$ was generated according to the statistical model \eqref{eqn:statistical_model}, that \(\mathbb{P}_n\) satisfies Assumptions \ref{asn: assumption high-level OLS ridge} and \ref{asn: stationary-prediction-law}, and that the stationary distribution $\mathbb{P}$ satisfies $\mathbb{E}_{\mathbb{P}}[x_i \epsilon_i] = 0_{k \times 1}$ and $\mathbb{E}_{\mathbb{P}}[x_i x_i^{\top}] = \Sigma$. If
\begin{enumerate}
\item[(i)] $\lambda_n/n \overset{p}{\to} \lambda \in [0,\infty)$ and $\sqrt{n}(\beta_n-\beta_0)\to b\in\mathbb{R}^k$,
\item[(ii)] there exist estimators $(\widehat{\Sigma},\widehat{\Omega})$ that are consistent for $(\Sigma,\Omega)$,
\item[(iii)] there exists $\delta>0$ such that
\[\sup_n \mathbb{E}_{\mathbb{P}_n}\!\left[(Z_n^\top \Sigma Z_n)^{1+\delta}\right] < \infty, \qquad Z_n \equiv \sqrt{n}(\widehat{\beta}_{\lambda_n}-\beta_n), \]
\end{enumerate}
then
\begin{equation} \label{eqn:excess_risk_approximation}
R(\widehat{a}_{\lambda_n};\beta_n,\mathbb P_n)-\sigma^2 = R^e_n(\lambda;b,\widehat{\Sigma},\widehat{\Omega}) + o_{(\beta_n,\mathbb P_n)}(1/n),
\end{equation}
where $\sigma^2 \equiv \mathbb{E}_{\mathbb{P}}[\epsilon^2]$.
\end{theorem}
\begin{proof}
See Appendix \ref{subsec: approximation to Ridge prediction risk}
\end{proof}

The difference between the finite sample prediction risk of a given data-driven predictor $\widehat{a}$---which we have denoted by $R(\widehat{a}_n;\beta_n,\mathbb{P}_n)$---and the residual variance $\sigma^2$ is typically referred to as the \emph{excess prediction risk} or simply \emph{excess risk}.\footnote{See, for example, \cite{mourtada2022exact}.} Theorem \ref{theorem: approximation to the prediction risk of ridge} says that, in large samples, the excess prediction risk of ridge regression equals---up to a term that is small in probability---by the approximate excess prediction risk function defined in \eqref{eqn: approximation to the prediction risk of ridge},  evaluated at $(\lambda,b,\widehat{\Sigma},\widehat{\Omega})$. The result shows that if we fix $(\widehat{\Sigma},\widehat{\Omega})$, the excess risk is expected to vary as a function of i) how large is the penalty parameter relative to the sample size ($\lambda$ is the probability limit of $\lambda_n/n$) and ii) how close is the reference vector $\beta_0$ to the true coefficient $\beta_n$ ($b$ is the limit of $\sqrt{n}(\beta_n-\beta_0)$). 

We note that part ii) of the assumptions of Theorem \ref{theorem: approximation to the prediction risk of ridge} can easily be verified in low-dimensional environments, where $k$ is fixed and $n$ grows to infinity. However, it is well known that---even if conditional homoskedasticity holds---the consistent estimation of $\Omega = \sigma^2_{\epsilon} \Sigma$ is a nontrivial problem when $k$ grows proportionally to $n$; see, for example, \cite{liu2020estimation}. The estimation of $\Omega$ in high-dimensional problems where heteroskedasticity is present is even more challenging; see \cite{cattaneo2018inference} and the references therein.

Theorem \ref{theorem: approximation to the prediction risk of ridge} is conceptually related to an active area of research in the statistics literature providing asymptotic and nonasymptotic approximations to the excess prediction risk of the ridge regression estimator; see, for example, \cite{dobriban2018high},\cite{hastie2022surprises}, \cite{mourtada2022elementary} and the references therein. All these papers make distributional assumptions that preclude data generating processes that exhibit conditional heteroskedasticity, time-series autocorrelation, or a nonlinear conditional expectation function. The additional assumptions in these papers, however, allow them to consider \emph{high-dimensional} asymptotics where the number of covariates is allowed to grow relative to the sample size. In contrast, all the analysis in this paper pertains a \emph{low-dimensional} model where we treat the number of covariates as fixed, and let the sample size diverge to infinity.

\subsection{Selection of $\lambda$}  

In this section, we use the approximation to the excess prediction risk of ridge regression given in Theorem \ref{theorem: approximation to the prediction risk of ridge} to make a concrete recommendation regarding the selection of the regularization parameter $\lambda_n$. The key insight is to pick $\lambda$---which is the probability limit of $\lambda_n/n$---to optimize the function $R^e_n(\lambda;b,\widehat{\Sigma},\widehat{\Omega})$, which is the leading term in \eqref{eqn:excess_risk_approximation}. The main conceptual challenge is that approximation depends on the unknown localization parameter $b \in \mathbb{R}^{k}$, which was defined to be the limit of the sequence $\sqrt{n}(\beta_n-\beta_0)$. Since $b$ cannot be estimated consistently, we suggest handling this \emph{localization} parameter in a way that we think is entirely analogous to the evaluation of risk functions in statistical decision theory: optimizing either the worst-case or the average risk. Our suggested method provides a well-defined mapping from the data (through the estimated covariance matrices $\widehat{\Sigma}$ and $\widehat{\Omega}$) to values of the regularization parameter. We provide details below.        

\subsubsection{Choosing $\lambda$ to minimize the worst-case approximate excess risk} \label{subsec:minimax}

We first consider the problem of minimizing the worst-case approximate excess risk. Let $\mathcal{B} \subseteq \mathbb{R}^{k}$ be a user-specified set of potential values for $b$. For example, one could take  $\mathcal{B} \equiv \{ b \in \mathbb{R}^{k} :  \|b\| \leq B\}$. for some known $B >0 $, where $\| \cdot \|$ denotes the standard Euclidean norm. We say that $\widehat{\lambda}^*_n$ is a $\mathcal{B}$-minimax selection of the regularization parameter in ridge regression if $\widehat{\lambda}^*_n  = n \cdot \widehat{\lambda}^* $ where $\widehat{\lambda}^*$ solves the minimax problem
\begin{equation} \label{eqn:minimax_choice_lambda}  
\adjustlimits \inf_{\lambda \geq 0 }  \sup_{ b \in \mathcal{B} } R^e_n(\lambda;b,\widehat{\Sigma},\widehat{\Omega}).
\end{equation}
We make two comments about the minimax problem in \eqref{eqn:minimax_choice_lambda}. First, if the set $\mathcal{B}$ is unbounded, then the $\mathcal{B}$-minimax choice of the regularization parameter in ridge regression is $\widehat{\lambda}_n = 0$. This means that in order to obtain a different choice of $\widehat{\lambda}_n$ under the minimax criterion we will require a bound on the parameter $b \in \mathbb{R}^k$. 

Second, while the outer optimization problem in \eqref{eqn:minimax_choice_lambda} could be handled via grid search (as it is done whenever $\widehat{\lambda}_n$ is chosen by means of cross-validation), the inner optimization problem requires the evaluation of the worst-case approximate excess risk.  Since only the first term in the expression of $R^e_n(\lambda;b,\widehat{\Sigma},\widehat{\Omega})$ in \eqref{eqn: approximation to the prediction risk of ridge} depends on $b$, we can solve for the worst-case risk by solving the following constrained optimization problem:
\begin{equation} \label{eqn: max approximate risk} 
\widehat{v}_{\textrm{worst-case}}(\lambda) \equiv \sup_{b} b^{\top}  (\widehat{\Sigma} + \lambda \mathbb{I}_k)^{-1} \widehat{\Sigma}  (\widehat{\Sigma} + \lambda \mathbb{I}_k)^{-1} b \quad \textrm{ subject to } b \in \mathcal{B}. 
\end{equation}
Algebra shows that when $\mathcal{B} \equiv \{ b \in \mathbb{R}^{k} : \|b\| \leq B\}$---for some known $B >0$, and with $\| \cdot \|$ given by the Euclidean norm---the constrained optimization problem in \eqref{eqn: max approximate risk} can be solved in closed-form up to a maximum eigenvalue computation (denoted by $\mu_{\textrm{max}}(\cdot)$):
\[ \widehat{v}_{\textrm{worst-case}} (\lambda) = B^2 \cdot \mu_{\textrm{max}}\left(  (\widehat{\Sigma} + \lambda \mathbb{I}_k)^{-1} \widehat{\Sigma}  (\widehat{\Sigma} + \lambda \mathbb{I}_k)^{-1}  \right). \]
Thus, in the special case in which the Euclidean norm of $b$ is bounded by $B$, the $\mathcal{B}$-minimax choice of $\widehat{\lambda}_n$ can be implemented as $\widehat{\lambda}^*_n = n \times \widehat{\lambda}^*$ where $\widehat{\lambda}^*$ solves the problem  
\begin{equation} \label{eqn: mminimax lambda general case}
\inf_{\lambda \in \mathbb{R}_{+}} \lambda^2 B^2 \mu_{\textrm{max}}\left( (\widehat{\Sigma} + \lambda \mathbb{I}_k)^{-1} \widehat{\Sigma}  (\widehat{\Sigma} + \lambda \mathbb{I}_k)^{-1}   \right) + \textrm{trace} \left(  (\widehat{\Sigma} + \lambda \mathbb{I}_k)^{-1} \widehat{\Omega}  (\widehat{\Sigma} + \lambda \mathbb{I}_k)^{-1} \widehat{\Sigma} \right).
\end{equation}
The minimization problem in \eqref{eqn: mminimax lambda general case} is a nonlinear (and nonconvex) optimization problem over $\mathbb{R}_{+}$ and, in general, does not have a closed-form solution. Our suggestion is to minimize this function over the same grid of parameter values used in the standard implementation of (leave-one-out) cross-validation for ridge regression provided in standard statistical packages such as \texttt{glmnet}.

\subsubsection{Choosing $\lambda$ to minimize average approximate excess risk} \label{subsec:average}
We now consider the problem of choosing the regularization parameter of ridge regression by minimizing the \emph{average} approximate excess risk. Let $\pi$ denote a user-specified probability distribution over $\mathbb{R}^{k}$. We say that $\widehat{\lambda}^*_n$ is a $\pi$-optimal selection of the regularization parameter in ridge regression if $\widehat{\lambda}^*_n  = n \cdot \widehat{\lambda}^* $ where $\widehat{\lambda}^*$ solves the problem
\begin{equation} \label{eqn:pi_choice_lambda}  
\inf_{\lambda \geq 0 } \mathbb{E}_{b \sim \pi} [  R^e_n(\lambda;b,\widehat{\Sigma},\widehat{\Omega})]. 
\end{equation}
Once again, since only the first term in the expression of $R^e_n(\lambda;b,\widehat{\Sigma},\widehat{\Omega})$ in \eqref{eqn: approximation to the prediction risk of ridge} depends on $b$, we can solve for the average excess risk by evaluating the following expectation:
\begin{equation} \label{eqn: average approximate risk} 
\widehat{v}_{\textrm{average}}(\lambda) \equiv \mathbb{E}_{b \sim \pi} \left [ b^{\top}  (\widehat{\Sigma} + \lambda \mathbb{I}_k)^{-1} \widehat{\Sigma}  (\widehat{\Sigma} + \lambda \mathbb{I}_k)^{-1} b \right]. 
\end{equation}
This expectation can be evaluated analytically in some cases. For example, suppose $b \sim \mathcal{N}_{k} \left( \mathbf{0}, C^2 \mathbb{I}_k \right)$, where $C$ is a user-specified hyper-parameter. Algebra shows that  
\[ \widehat{v}_{\textrm{average}} (\lambda) = C^2 \cdot \textrm{trace}\left(  (\widehat{\Sigma} + \lambda \mathbb{I}_k)^{-1} \widehat{\Sigma}  (\widehat{\Sigma} + \lambda \mathbb{I}_k)^{-1}  \right). \]
A similar formula can be derived if the distribution over $b$ is elliptical with mean zero and a covariance matrix proportional to the identity matrix; see Appendix \ref{subsec:elliptical_distributions}. The $\pi$-optimal choice of $\widehat{\lambda}_n$ in these cases is $\widehat{\lambda}^*_n = n \times \widehat{\lambda}^*$, where $\widehat{\lambda}^*$ solves the problem  
\begin{equation} \label{eqn: average lambda general case}
\inf_{\lambda \in \mathbb{R}_{+}} \lambda^2 C^2 \textrm{trace}\left( (\widehat{\Sigma} + \lambda \mathbb{I}_k)^{-1} \widehat{\Sigma}  (\widehat{\Sigma} + \lambda \mathbb{I}_k)^{-1}   \right) + \textrm{trace} \left(  (\widehat{\Sigma} + \lambda \mathbb{I}_k)^{-1} \widehat{\Omega}  (\widehat{\Sigma} + \lambda \mathbb{I}_k)^{-1} \widehat{\Sigma} \right).
\end{equation}
Just as for the $\mathcal{B}$-minimax choice of regularization parameter, our suggestion is to minimize this function over the same grid of parameter values used in the standard implementation of (leave-one-out) cross-validation for ridge regression provided in standard statistical packages.

\subsubsection{Isotropic Features} 
We now specialize the derivations in Sections \ref{subsec:minimax} and \ref{subsec:average} to the case in which $\Sigma$---the matrix of second moments of the vector of covariates---is known to be  proportional to identity matrix; that is, $\Sigma = \sigma^2_{x} \mathbb{I}_k$. \cite{hastie2022surprises} refer to predictions problems of this form as prediction with \emph{isotropic features}.\footnote{\cite{hastie2022surprises} further assume that $\sigma^2_{x}=1$.} Algebra shows that with isotropic features, both the minimax and average-risk objectives outlined in the previous subsections simplify substantially, allowing explicit characterization of the optimal regularization parameter for ridge regression. 
The key observation is that under isotropic features, the approximate excess risk in \eqref{eqn: approximation to the prediction risk of ridge} simplifies to
\begin{equation}\label{eqn: isotropic oracle approximation net of sigma^2}
    R^e_n( \lambda; b,\sigma^2_x,\Omega ) = \frac{1}{n}\frac{\lambda^2\sigma^2_x}{(\sigma^2_x + \lambda)^2}b^{\top}b + \frac{1}{n}\frac{\sigma^2_x}{(\sigma^2_x + \lambda)^2} \operatorname{trace}(\Omega).
\end{equation}
The following theorem presents the optimal selection of the regularization parameter for ridge regression under the minimax and average risk criterion.

\begin{theorem}[Optimal choice of regularization parameter with isotropic features]\label{theorem: selection of lambda}
Suppose that the assumptions of Theorem \ref{theorem: approximation to the prediction risk of ridge} hold. Assume that $\Sigma=\sigma_x^2 \mathbb{I}_k, \: \sigma_x^2>0,$ and that there exist estimators $(\widehat{\sigma}_x^2,\widehat{\Omega})$ that are consistent for $(\sigma_x^2,\Omega)$. Then:

\begin{enumerate}
    \item[(1)] If $\mathcal{B} \equiv \{ b \in \mathbb{R}^{k} \: \ \:  \|b\| \leq B\}$---for some known $B>0$, and with $\| \cdot \|$ given by the Euclidean norm---the $\mathcal{B}$-minimax selection of the regularization parameter in ridge regression is
    \[\widehat{\lambda}_{n,\mathrm{minimax}}^*=n \cdot \frac{\operatorname{trace}(\widehat{\Omega})}{\widehat{\sigma}_x^2 B^2}.\]

    \item[(2)] If $\mathbb E_\pi[b^\top b]<\infty$ and $\mathbb E_\pi[b^\top b]>0$, the $\pi$-optimal selection of the regularization parameter in ridge regression is
    \[\widehat{\lambda}_{n,\pi}^* =n \cdot \frac{\operatorname{trace}(\widehat{\Omega})}{\widehat{\sigma}_x^2 \mathbb E_\pi[b^\top b]}.\]
     
\end{enumerate}
\end{theorem}
\begin{proof}
See Appendix \ref{subsec: selection of lambda}
\end{proof}

Since the introduction of the ridge regression estimator, the problem of providing a concrete recommendation for its regularization parameter has been analyzed in multiple papers; for example, see \cite{gibbons1981} and the references therein.  The idea of using an approximation of the risk function to choose regularization parameters has a long precedent in statistics and econometrics; see, for example the discussion of \cite{abadie2019choosing} on p. 744. To the best of our knowledge the formulae we provide in Theorem \ref{theorem: selection of lambda} are new to the literature. For the formula we provide in (2), the closest reference we were able to find in the literature is the optimal regularization parameter derived in Theorem 2.1 in \cite{dobriban2018high}, p. 254. Their recommendation is based on high-dimensional model where covariates and regression residuals are independent, but covariates need not be isotropic. Assuming a prior on regression coefficients such that $\mathbb{E}_{\pi}[\beta]=0$ and $\textrm{Var}_{\pi}[\beta] = \alpha^2 \mathbb{I}_k/k$ (see their Assumption RRC), \cite{dobriban2018high} recommend a regularization parameter of the form 
\[ \hat{\lambda} = n \cdot\frac{k}{n} \frac{1}{\alpha^2} = \frac{k}{\alpha^2}, \]
which assumes $\sigma^2_{\epsilon}=1$. In the homoskedastic case with isotropic covariates---$\Omega=\sigma^2_{\epsilon} \sigma^2_x \mathbb{I}_k$---our formula matches their recommendation since 
 \[\widehat{\lambda}_{n,\pi}^* =n \cdot \frac{\operatorname{trace}(\widehat{\Omega})}{\widehat{\sigma}_x^2 \mathbb E_\pi[b^\top b]} = \frac{k}{\mathbb{E}_{\pi}[\beta_n^{\top}\beta_n]} =  \frac{k}{\alpha^2}, \]
where we have used the fact that the reference vector $\beta_0$ is zero, that $\beta_n=b/\sqrt{n}$, and that $\sigma^2_{\epsilon}=1$. However, our result allows for heteroskedasticity and autocorrelation in the data generating process by adjusting the regularization parameter as a function of $\Omega$. The main limitation of our result, as discussed before, is that our theory is only applicable to low-dimensional models. As we discussed after the statement of Theorem \ref{theorem: approximation to the prediction risk of ridge}, it is not entirely clear how to construct a good estimator for $\Omega$ in high-dimensional models. While we think that the extension of our main theorems to high-dimensional environments could significantly extend the relevance and applicability of our results, such an extension is outside the scope of this paper (and will likely require different techniques).

More generally, it is important to mention again that one of the initial motivations of our paper was to recover an expression analogous to the optimal regularization parameter for the ridge regression estimator given by \cite{hastie2022surprises} in their Corollary 6 to Theorem 6. Their results show that using high-dimensional asymptotics where $k/n \rightarrow \gamma$, the oracle choice of regularization parameter for the ridge estimator (assuming $x_i$ and $\epsilon_i$ are independent and that $\Sigma = \mathbb{I}_k$) takes the form $\lambda^* = n \mathbb{E}[\epsilon_i^2]\gamma/\| \beta \|^2$. Assuming conditional homoskedasticity, our formulae in \eqref{eqn: isotropic-pi-optimal introduction} and \eqref{eqn: isotropic-minimax-optimal introduction} become analogous to \cite{hastie2022surprises}'s formula, but with the unknown $||\beta_n||^2$ replaced by either the prior mean of $\| b \|^2$ or the maximum value of $\|b\|^2$.

\section{Simulations}
\label{sec:simulations}

We now examine the extent to which the asymptotic approximations presented in Section \ref{sec:main results} capture the finite-sample behavior of the ridge estimator. For this purpose, we present Monte-Carlo simulations to calculate and compare the prediction risk of the following estimators: 
\begin{enumerate}
\item  The ridge regression with the regularization parameter chosen to minimize the worst-case approximate excess-risk criterion in \eqref{eqn: mminimax lambda general case}.
\item  The ridge estimator tuned according to the standard leave-one-out cross-validation; see Section 3 in \cite{patil2021uniform}.
\item The OLS estimator.
\end{enumerate}
Henceforth, we refer to these estimators as minimax ridge, LOO-CV ridge, and OLS, respectively. Recall that $\mathcal{N}_k$ denotes a multivariate Gaussian distribution of dimension $k$.

\subsection{DGP-1}
\label{subsec:dgp1}

We first consider the stylized Gaussian homoskedastic linear regression model
\[
y_i=x_i^\top\beta_n+\epsilon_i,
\qquad
x_i\sim \mathcal{N}_k(0,\sigma_x^2 \mathbb{I}_k),
\qquad
\epsilon_i\sim \mathcal{N}_k(0,\sigma^2),
\qquad
\beta_n=\beta_0+\frac{b}{\sqrt n},
\]
with $\beta_0= (0, \ldots,0)^\top$. We assume that \(\epsilon_i\) is independent of \(x_i\). Let \(\Sigma=\mathbb E[x_i x_i^\top]=\sigma_x^2 \mathbb{I}_k\), and \(\Omega=\mathbb E[x_i x_i^\top \epsilon_i^2]\). Under homoskedasticity and independence, \(\Omega=\sigma^2\Sigma=\sigma^2\sigma_x^2 \mathbb{I}_k\).

Throughout this subsection we set $k=10$, $\sigma_x^2=1$, $\sigma^2=1$, and $b=(1,\ldots,1)^\top$. We set $B\equiv \|b\|=\sqrt{10}$. In our simulations, we consider six different sample sizes: $n\in\{500,1000,1500,2000,$ $2500,3000\}$. Each simulation uses 2,000 Monte Carlo repetitions.

We use the term \emph{minimax ridge} to refer to the ridge estimator that uses the regularization parameter given below in \eqref{eqn: DGP-1 minimax} and $\beta_0$ as the reference vector. Since the vector of covariates has isotropic features ($\Sigma=\sigma_x^2 \mathbb{I}_k$), the regularization parameter is chosen as in Theorem \ref{theorem: selection of lambda}; that is:
\begin{equation}\label{eqn: DGP-1 minimax}
    \widehat{\lambda}_{n,\mathrm{minimax}} = n\frac{\operatorname{trace}(\widehat{\Omega})}{\widehat{\sigma}_x^2 B^2},
\end{equation}
where \(\widehat{\sigma}_x^2\) is estimated from the sample covariance matrix of the covariates, and \(\Omega\) is consistently estimated by an estimator of \(\mathbb{E}[x_i x_i^\top\epsilon_i^2]\), obtained by replacing $\epsilon_i$ with OLS residuals and imposing conditional homoskedasticity. 

For the LOO-CV (leave-one-out cross-validation) ridge, we use the leave-one-out cross-validation procedure as in \citet[Section 3]{patil2021uniform}.\footnote{We set \texttt{sklearn.linear\_model.RidgeCV} with \texttt{cv=None} and \texttt{fit\_intercept=False}.}
The candidate set of regularization parameters $\lambda_n=nc$ is based on a nonuniform grid of $500$ strictly positive candidate values of $c$ spanning $[0.001,5]$. We let \texttt{RidgeCV} choose the minimizer of the leave-one-out criterion over that grid.

For each sample size $n$ and $m = 1, \dots, 2000$, we proceed as follows:
\begin{enumerate}
    \item Draw an independent sample $D_n^{(m)}=\{(Y_i^{(m)},X_i^{(m)})\}_{i=1}^n$ from the Gaussian model.
    \item Calculate the  minimax ridge, LOO-CV ridge, and OLS using \(D_n^{(m)}\).
    \item Calculate the prediction error for each estimator \(\widehat\beta_n^{(m)}\), 
    \[ \mathbb{E}_{(Y,X)\sim (\beta_n,\mathbb{P})} \left[ (Y-X^\top \widehat\beta_n^{(m)})^2 \right] = \sigma^2+ (\widehat\beta_n^{(m)}-\beta_n)^\top \Sigma (\widehat\beta_n^{(m)}-\beta_n). \]
\end{enumerate}
We then approximate the prediction risk by averaging the prediction errors across repetitions:
\[\frac{1}{2000}
\sum_{m=1}^{2000}
\mathbb{E}_{(Y,X)\sim (\beta_n,\mathbb{P})}
\left[
(Y-X^\top \widehat\beta_n^{(m)})^2
\right].\]

\begin{figure}[t]
\centering
\includegraphics[width=0.72\textwidth]{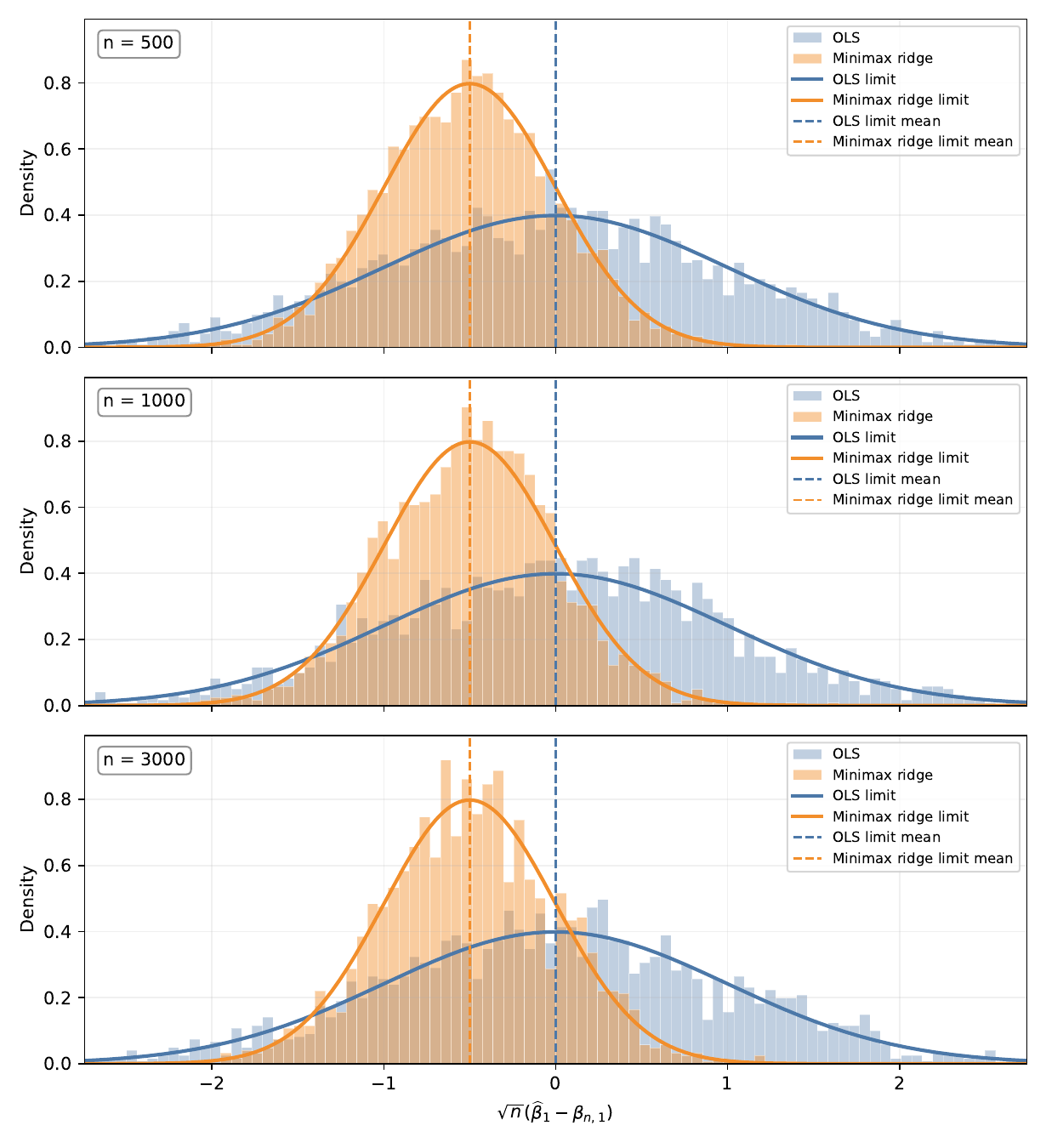}
\caption{Distribution of the first coefficient under DGP-1}
\label{fig:dgp1-z1-ols-minimax}
\end{figure}

Figure~\ref{fig:dgp1-z1-ols-minimax} reports the distribution of $\sqrt{n}(\widehat\beta_1-\beta_{n,1})$ for the OLS and minimax ridge estimator, with the Gaussian limits implied by Theorem~\ref{theorem: big lambda_n approximation}: $\mathcal{N}_1(0,\tfrac{\sigma^2}{\sigma_x^2})$ for OLS and $\mathcal{N}_1(-\lambda b_1/(\sigma_x^2+\lambda),\frac{\sigma^2\sigma_x^2}{(\sigma_x^2+\lambda)^2})$ for minimax ridge. The approximation is accurate already at $n=500$ and remains so as $n$ grows. Consistent with our asymptotic approximations, the minimax ridge distribution is shifted away from the origin, but markedly more concentrated than the OLS distribution. In contrast, OLS is correctly centered but presents a large dispersion.

We now compare the estimators using the excess prediction risk; that is, $R(\widehat{a}_{\lambda_n};\beta_n,\mathbb P_n) - \sigma^2$. 
Figure~\ref{fig:dgp1-risk} compares the excess prediction risk of the three estimators relative to minimax ridge estimator across different sample sizes. The figure suggests that, in this design, our recommended minimax ridge outperforms both LOO-CV ridge and OLS. Across sample sizes, the risk of LOO-CV ridge is about 20 percent higher than that of minimax ridge.

\begin{figure}[h!]
\centering
\includegraphics[width=0.66\textwidth]{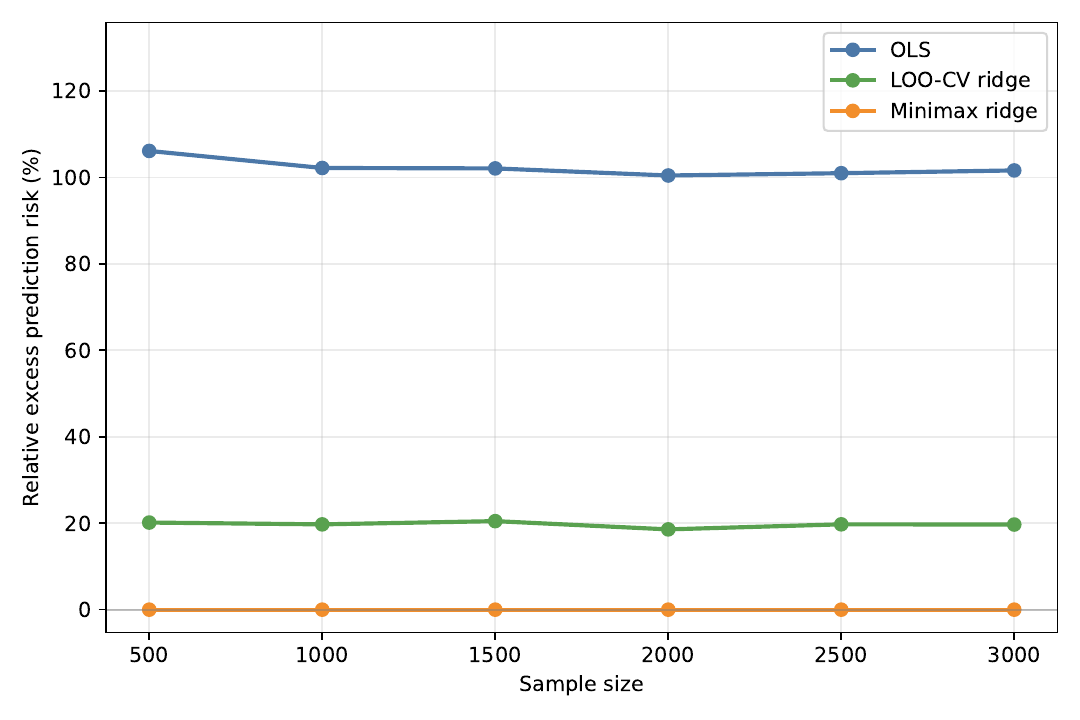}
\caption{Relative excess prediction risk under DGP-1}
\label{fig:dgp1-risk}
\end{figure}

We next evaluate our approximation to the excess prediction risk of ridge estimators presented in Theorem \ref{theorem: approximation to the prediction risk of ridge}. 
Figure~\ref{fig:dgp1-risk-function} displays a comparison between the Monte Carlo estimate of excess prediction risk $R(\widehat{a}_{\lambda_n};\beta_n,\mathbb P_n)-\sigma^2$ (solid orange line) and our \emph{feasible} approximate excess-risk $R_n^e(\lambda;b,\widehat{\Sigma},\widehat{\Omega})$ (shaded region) for several values of the scaled regularization parameter $\lambda_n/n \in [0.001,5]$.\footnote{For each value of $\lambda$, the approximate excess-risk $R_n^e(\lambda;b,\widehat{\Sigma},\widehat{\Omega})$ is data-dependent; therefore, for each simulation we have a different value. To capture its range of values, we report the central 95 percent range after dropping the lowest 2.5 percent and highest 2.5 percent of values for each candidate ratio.} In all the panels, we scale the vertical axis by the sample size $n$ and report results consistent with \eqref{eqn:excess_risk_approximation}, showing that the approximate excess-risk calculation tracks the excess-risk curve closely. Figure~\ref{fig:dgp1-risk-function} also presents the approximate excess risk minimizer (vertical dashed line), which is the median (across repetitions) of $\lambda_n/n$ based on \eqref{eqn: DGP-1 minimax}. For all three sample sizes ($n=500$, 1,000, and 3,000), the median selected tuning ratio is approximately 0.99, which is close to the theoretical value $\lambda_n/n=1$.

\begin{figure}[h!]
\centering
\includegraphics[width=0.72\textwidth]{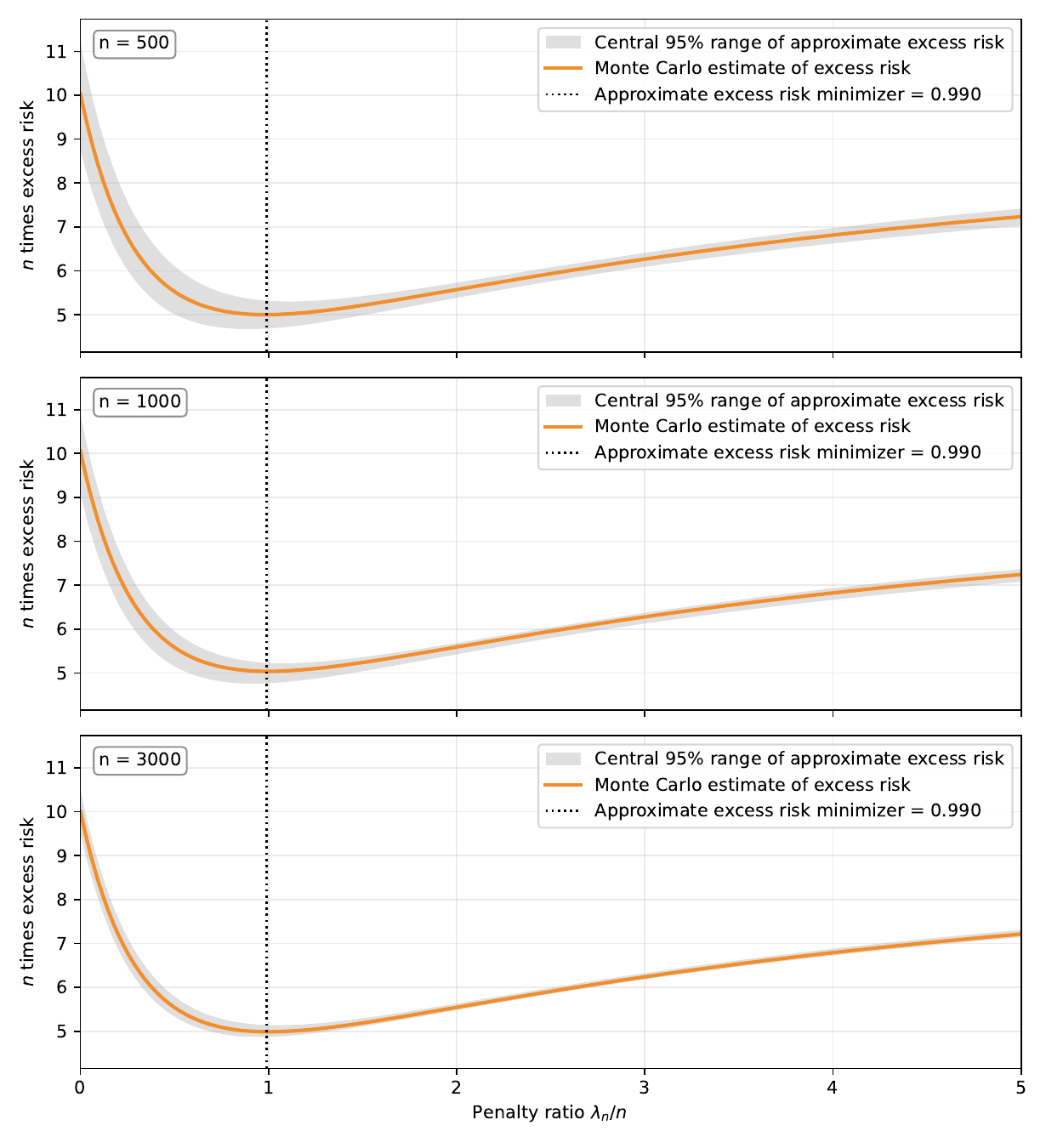}
\caption{Risk curves under DGP-1}
\label{fig:dgp1-risk-function}
\end{figure}

We also compare the scaled regularization parameter ($\widehat\lambda_n/n$) used by the LOO-CV and minimax ridge estimators. Figure~\ref{fig:dgp1-ratio} in Appendix~\ref{subsec: additional simulation results} confirms that the scaled regularization parameter of the minimax ridge---based on \eqref{eqn: DGP-1 minimax}---is close to $1$, which is the limit of $\lambda_n/n$, as expected. By contrast, the scaled regularization parameter of the LOO-CV ridge is widely dispersed across the values it can take.

\subsection{DGP-2}
\label{subsec:dgp2}

We next consider a low-dimensional version of the weak-signal design based on \citet[Section~3.1]{shen2025weak}. Our asymptotic theory in Section~\ref{sec:main results} considers models in which the number of covariates is small relative to the sample size. Thus, before considering a high-dimensional version of this design, it is useful to first understand the performance of our approximations in a non-isotropic design with a smaller number of covariates. We note that this design (and also the one presented in the next section) implies a minor departure from Assumption \ref{asn: stationary-prediction-law} (strict stationarity) by allowing heterogeneous variances. 

Following \cite{shen2025weak}, we let the outcome be generated by the linear model $y_i=x_i^\top\beta_n+\epsilon_i$, with $n=500$ and $k=50$. The ridge reference vector is $\beta_0=0_k$, so the local parameter is $b = \sqrt{n} (\beta_n - \beta_0) = \sqrt n\beta_n$. The design matrix is generated as $X=\Sigma_1^{1/2}Z\Sigma_2^{1/2}$, where $Z \in \mathbb{R}^{n \times k}$ has i.i.d. standard normal entries, $\Sigma_1 \in \mathbb{R}^{n \times n}$ controls the dependence across observations, and $\Sigma_2 \in \mathbb{R}^{k\times k}$ controls the dependence across regressors. For $\Sigma_1$, we follow \cite{shen2025weak} and set $(\Sigma_1)_{ij}=2^{-|i-j|}$ for $1 \leq i,j \leq n$. For $\Sigma_2$, we also follow \cite{shen2025weak} and construct the second moments of covariates by orthogonal diagonalization so that $\Sigma_2=Q\Lambda Q^\top$. The orthogonal matrix $Q$ is randomly drawn once, the eigenvalues in $\Lambda$ are drawn once from $U[0.5,1.5]$, and $\Sigma_2$ is then held fixed throughout the simulation. We note that, by construction, the matrix $\Sigma_2$ is not isotropic. In Appendix~\ref{subsec:dgp2-diagnostics} we describe the distribution of eigenvalues of this matrix and the coefficient vector used in the simulation. Since the diagonal entries of $\Sigma_1$ are equal to one, we can show that the population prediction covariance is $\Sigma=\mathbb E[x_i x_i^\top]=\Sigma_2$. The errors are independent of the regressors and are generated with diagonal heteroskedasticity, with the diagonal entries of the error-variance matrix drawn once from $U[0.5,1.5]$ and held fixed ($\sigma_{\epsilon,i}^2$). Thus, the asymptotic variance $\Omega$ in Assumption \ref{asn: assumption high-level OLS ridge} is $\Omega
=\bar\sigma_\epsilon^2\Sigma_2$, where $\bar{\sigma}^2_{\epsilon}\equiv n^{-1}\sum_{i=1}^{n}\sigma_{\epsilon,i}^2$, which is close to $1$ in this design consistent with the strong law of large numbers.

Following the coefficient-generation model in \citet[Section~3.1]{shen2025weak}, we first draw a preliminary coefficient vector $\widetilde\beta = (\widetilde\beta_1\ldots,\widetilde\beta_k)$, with independent coordinates satisfying $\widetilde\beta_j = 0$ with probability $0.2$ and $\widetilde\beta_j \sim \mathcal{N}_1(0,1.25)$ with probability $0.8$. For each target value of \(R^2\) ($5\%$, $20\%$, and $50\%$), we rescale this same draw to obtain the corresponding true coefficient vector \(\beta_n\). The resulting \(\beta_n\) is then held fixed across the $2,000$ Monte Carlo repetitions.

For the minimax criterion, we calibrate the radius $B$ from the same coefficient-generation scheme: for each target $R^2$, we draw 2,000 coefficient vectors, rescale each draw to the target $R^2$, compute $\sqrt {n}\|\beta_n\|$, and use the 90th percentile as the baseline value of $B$. The $\sqrt n$ normalization is the one used in the local parameter $b=\sqrt n\beta_n$. To accommodate the finite-sample effect of estimating a non-negligible number of coefficients, we adjust the approximate excess risk function by multiplying the variance term by $n/(n-k)$. Since this design is not isotropic, minimax ridge uses the general non-isotropic approximate excess-risk criterion implied by Theorem~\ref{theorem: approximation to the prediction risk of ridge}. In each Monte Carlo repetition, we draw a new sample $(X^{(m)},Y^{(m)})$, and we plug in the sample analogues $\widehat\Sigma_{\mathrm{df}}^{(m)}=X^{(m)\top}X^{(m)}/(n-k),\quad \widehat\sigma_{\epsilon,\mathrm{df}}^{2,(m)}
=\frac{1}{n-k}\sum_{i=1}^n\widehat u_{i,\mathrm{OLS}}^{(m)2}$, and $\widehat\Omega_{\mathrm{df}}^{(m)}
=\widehat\sigma_{\epsilon,\mathrm{df}}^{2,(m)}
\widehat\Sigma_{\mathrm{df}}^{(m)}$. For each target $R^2$ and each Monte Carlo repetition, we follow the same steps as in DGP-1.   

In each Monte Carlo repetition \(m\),  we choose  $\widehat\lambda_{n,\mathrm{minimax}}^{(m)} = (n-k) \times \widehat{\lambda}^{*,(m)}$, where  $\widehat{\lambda}^{*,(m)}$ minimizes, over a grid of candidate values, the sum of \eqref{eqn:dgp2_df_minimax_lambda} and \eqref{eqn:dgp2_df_minimax_lambda_variance} defined below:
\begin{equation}
\label{eqn:dgp2_df_minimax_lambda}
\lambda^2 B^2
\mu_{\textrm{max}}
\left(
(\widehat\Sigma_{\mathrm{df}}^{(m)}+\lambda \mathbb{I}_k)^{-1}
\widehat\Sigma_{\mathrm{df}}^{(m)}
(\widehat\Sigma_{\mathrm{df}}^{(m)}+\lambda \mathbb{I}_k)^{-1}
\right),
\end{equation}
\begin{equation}\label{eqn:dgp2_df_minimax_lambda_variance} 
\frac{n}{n-k}
\operatorname{trace}
\left[
(\widehat\Sigma_{\mathrm{df}}^{(m)}+\lambda \mathbb{I}_k)^{-1}
\widehat\Omega_{\mathrm{df}}^{(m)}
(\widehat\Sigma_{\mathrm{df}}^{(m)}+\lambda \mathbb{I}_k)^{-1}
\widehat\Sigma_{\mathrm{df}}^{(m)}
\right].
\end{equation}
Relative to the criterion in \eqref{eqn: mminimax lambda general case}, the only modification consists of replacing the plug-in covariance matrices by their degrees-of-freedom adjusted analogues and multiplying the variance component by $n/(n-k)$. After the minimization, the reported ridge tuning parameter is $\widehat\lambda_{n,\mathrm{minimax}}^{(m)}$, so all selected tuning ratios are reported on the scale $\widehat\lambda_n/n$.

For the LOO-CV ridge estimator, we use the same implementation as in DGP-1.\footnote{\texttt{sklearn.linear\_model.RidgeCV} with \texttt{cv=None} and \texttt{fit\_intercept=False}.} The number of Monte Carlo repetitions is 2,000 for each target value of \(R^2\).

\begin{figure}[t]
\centering
\includegraphics[width=0.72\textwidth]{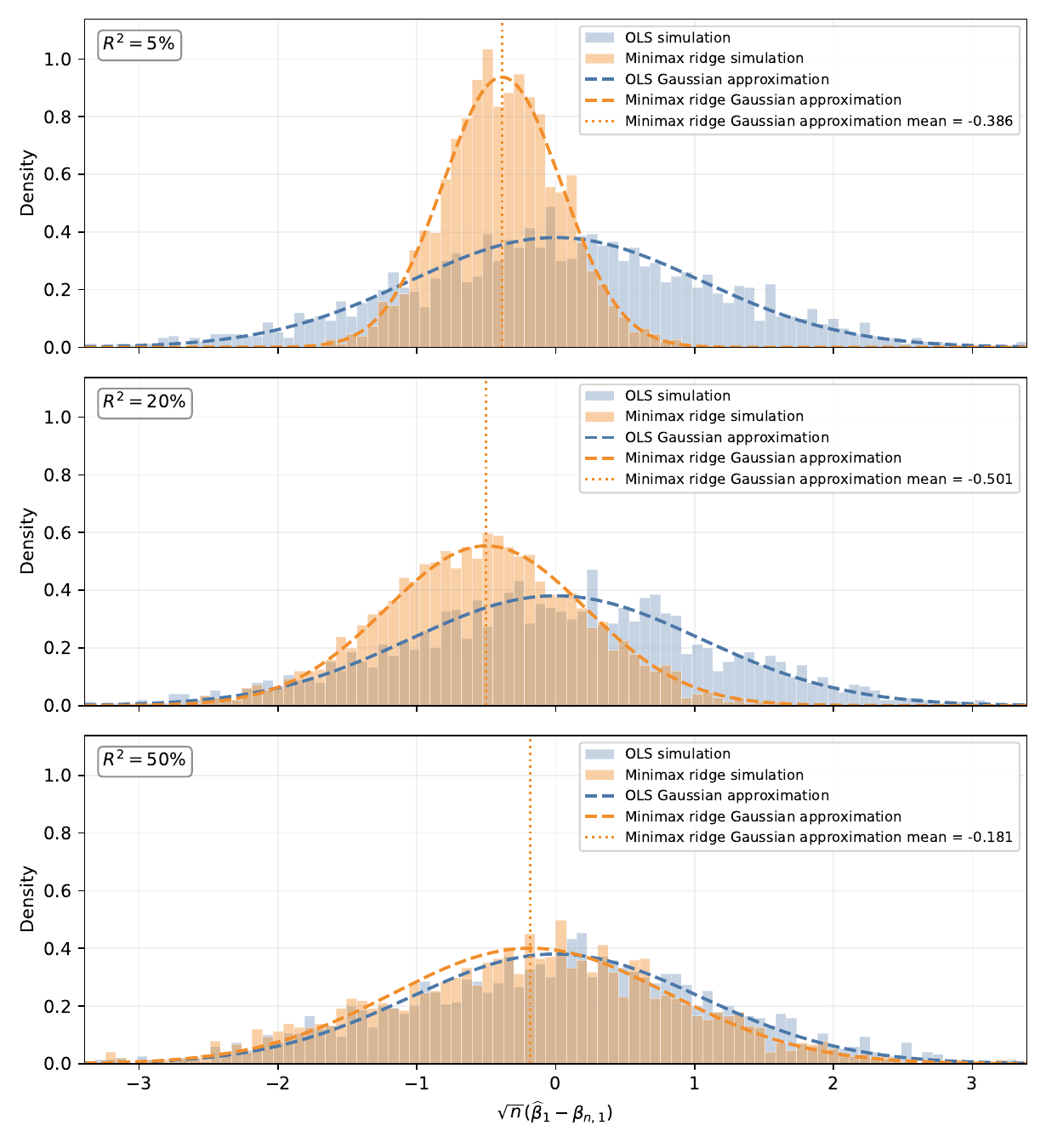}
\caption{Distribution of the first coefficient under DGP-2}
\label{fig:dgp2-z1-ols-minimax}
\end{figure}

The following figures report our simulation results. Figure~\ref{fig:dgp2-z1-ols-minimax} reports the simulated distribution of $\sqrt{n}(\widehat\beta_1-\beta_{n,1})$ for OLS and minimax ridge, together with the Gaussian approximations implied by our theory. The Gaussian curves are computed using the population matrices in the simulation, rather than fitted to the Monte Carlo histograms. Similar to the pattern observed in DGP-1, comparing with the OLS results, the distribution of the minimax ridge estimator is slightly biased but has a smaller variance.

\begin{figure}[h!]
\centering
\includegraphics[width=0.66\textwidth]{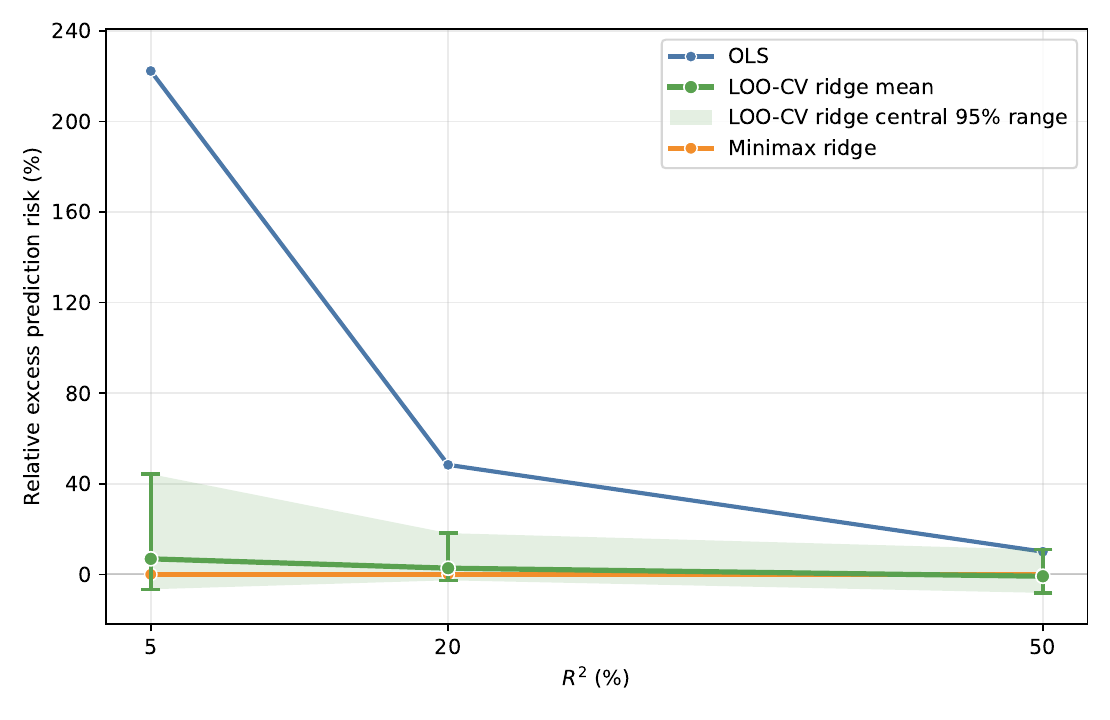}
\caption{Relative excess prediction risk under DGP-2, LOO-CV and minimax}
\label{fig:dgp2-risk}
\end{figure}

We next compare the three estimators using excess prediction risk. Figure~\ref{fig:dgp2-risk} reports the excess prediction risk of OLS and LOO-CV ridge relative to minimax ridge. The shaded region reports the central 95 percent range of the repetition-specific relative excess risk of LOO-CV ridge. The figure shows that OLS is dominated by both ridge estimators, especially when $R^2$ is small. LOO-CV ridge and minimax ridge have very similar excess prediction risk: LOO-CV is slightly above minimax ridge for the lower two values of $R^2$ and slightly below minimax ridge when $R^2=50\%$.

\begin{figure}[h!]
\centering
\includegraphics[width=0.72\textwidth]{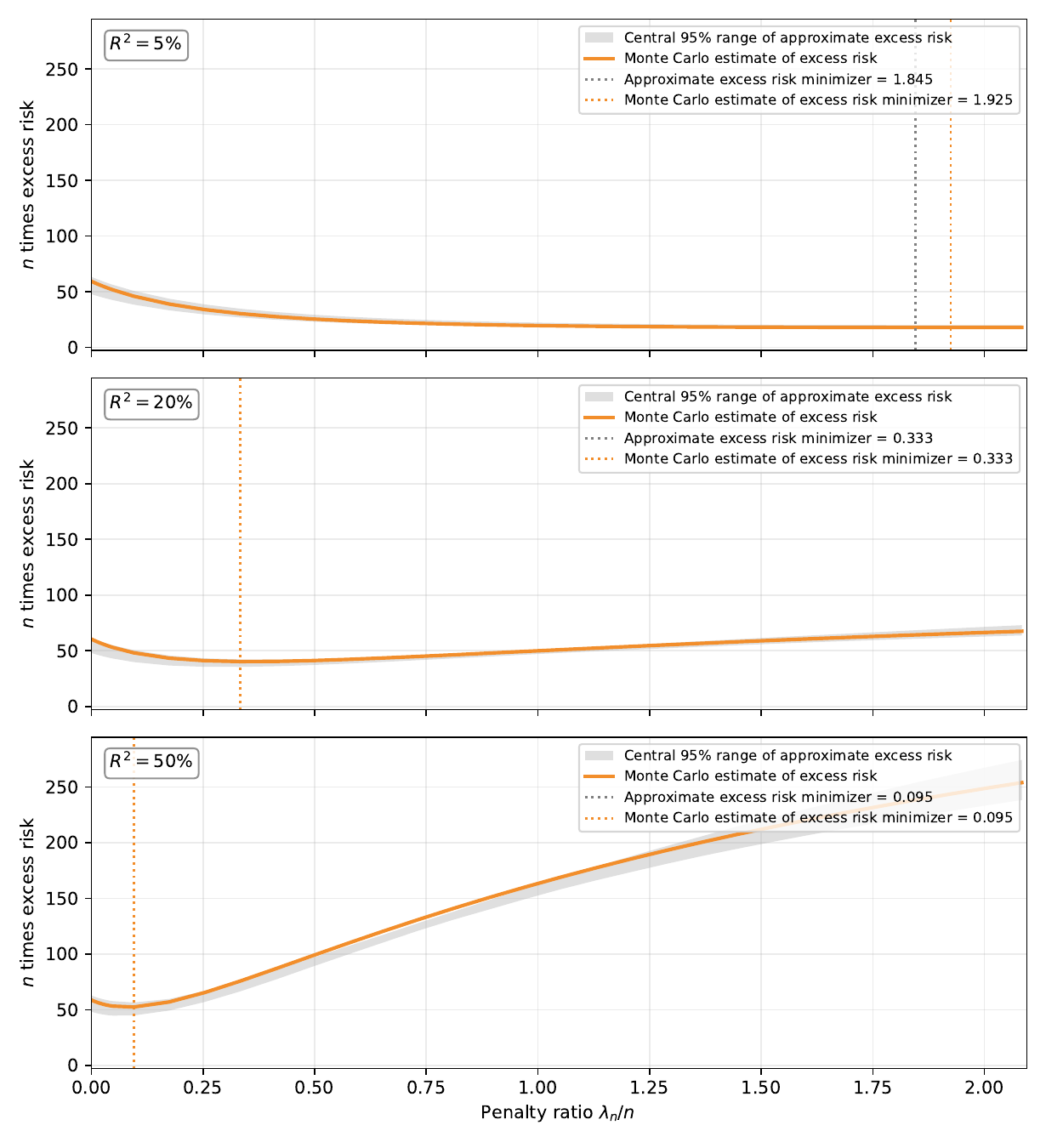}
\caption{Risk curves under DGP-2}
\label{fig:dgp2-risk-function}
\end{figure}

We next evaluate our approximation to the excess prediction risk of ridge estimators presented in Theorem~\ref{theorem: approximation to the prediction risk of ridge} under DGP-2. Figure~\ref{fig:dgp2-risk-function} displays a comparison between the Monte Carlo estimate of the scaled excess prediction risk, $n\left\{R(\widehat a_{\lambda_n};\beta_n,\mathbb P_n)-\sigma^2\right\}$ (solid orange line), and our scaled \emph{feasible} approximate excess risk $n R_n^e(\lambda;b,\widehat{\Sigma},\widehat{\Omega})$ (shaded region).  To capture the range of the feasible approximation (across repetitions), we report the central 95 percent range after dropping the lowest 2.5 percent and highest 2.5 percent of values for each candidate tuning ratio. In this lower-dimensional version of the design, the feasible approximate excess-risk calculation tracks the Monte Carlo risk curve closely, including the location of the low-risk region. Appendix~\ref{subsec:dgp2-diagnostics} reports additional diagnostics for the eigenvalues of $\Sigma_2$, the coefficient vectors, and the selected tuning ratios.

\subsection{DGP-3}
\label{subsec:dgp3}

We next consider the high-dimensional version of the weak-signal design based on Section~3.1 of \citet{shen2025weak}. Our asymptotic theory focused on models in which the number of covariates is small relative to the sample size. Thus, it is of interest to understand the performance of our approximations in models where the number of covariates are roughly of the same order as the sample size. The setup is exactly the same as the setup in DGP-2 in the previous subsection, except here we set $k=300$. 

\begin{figure}[t]
\centering
\includegraphics[width=0.72\textwidth]{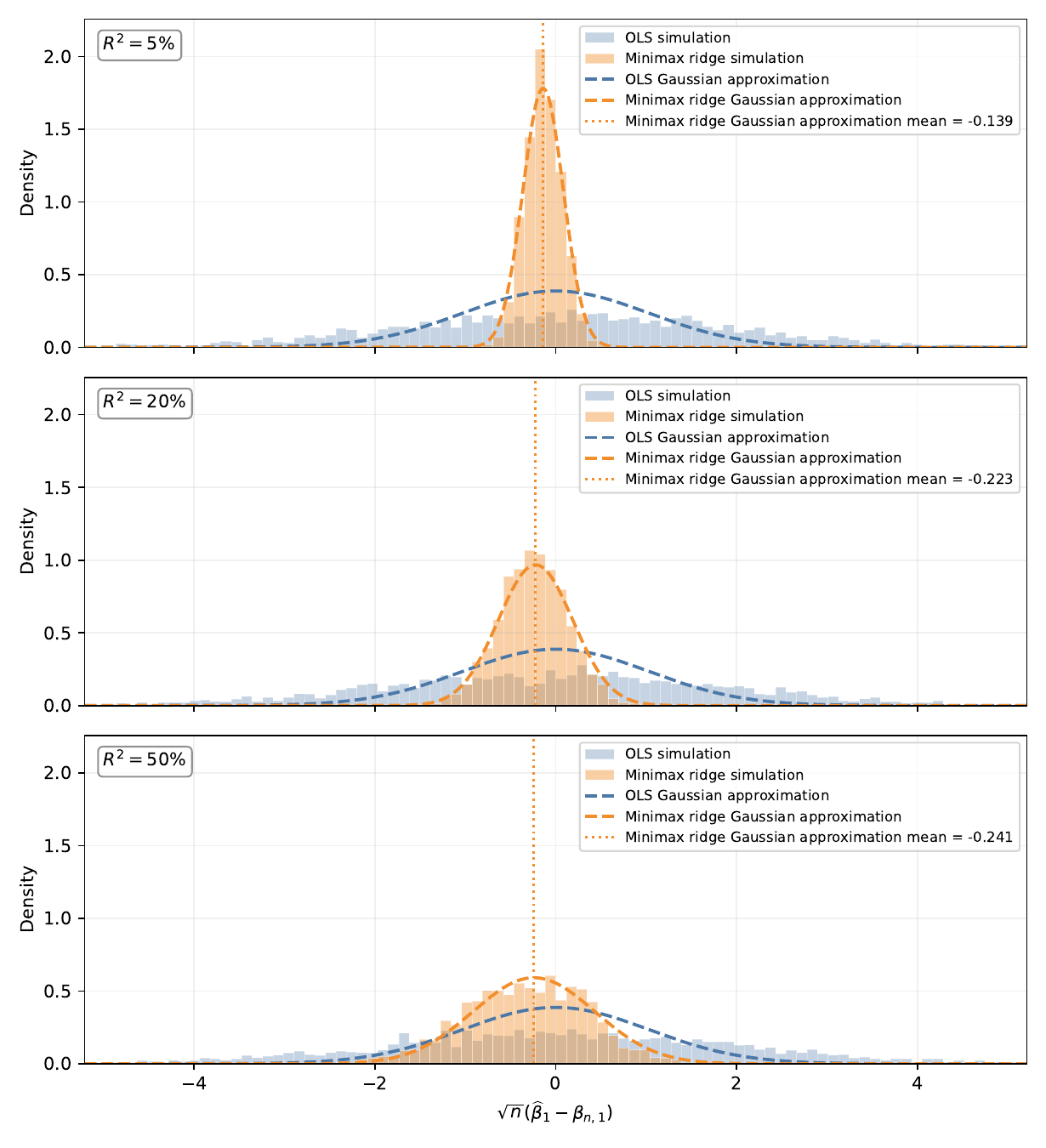}
\caption{Distribution of the first coefficient under DGP-3}
\label{fig:dgp3-z1-ols-minimax}
\end{figure}

The following figures report our simulation results. Figure~\ref{fig:dgp3-z1-ols-minimax} reports the simulated distribution of $\sqrt{n}(\widehat\beta_1-\beta_{n,1})$ for OLS and minimax ridge, together with the Gaussian approximations implied by our theory. The Gaussian curves are computed using the population matrices in the simulation, rather than fitted to the Monte Carlo histograms. Similar to the pattern observed in DGP-1 and DGP-2, comparing with the OLS results, the distribution of the minimax ridge estimator is biased by shrinkage but has a smaller variance.

\begin{figure}[h!]
\centering
\includegraphics[width=0.66\textwidth]{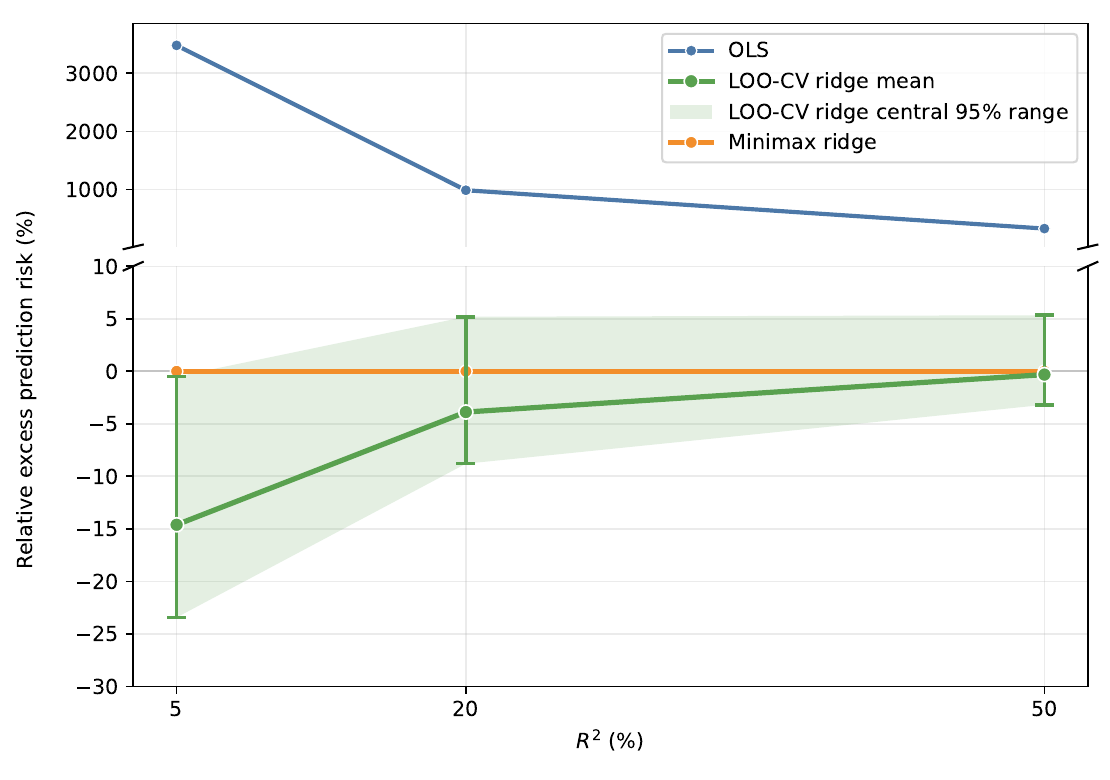}
\caption{Relative excess prediction risk under DGP-3, LOO-CV and minimax}
\label{fig:dgp3-risk}
\end{figure}

We next compare the three estimators using excess prediction risk. Figure~\ref{fig:dgp3-risk} reports the excess prediction risk of OLS and LOO-CV ridge relative to minimax ridge. The shaded region reports the central 95 percent range of the repetition-specific relative excess risk of LOO-CV ridge. The figure shows that OLS is dominated by both ridge estimators, especially when $R^2$ is small. LOO-CV ridge and minimax ridge have comparable excess prediction risk, although LOO-CV ridge has slightly lower mean excess risk than minimax ridge in this high-dimensional design. Thus, in the regime where $k/n$ is large, the approximation-based rule should not be interpreted as uniformly improving on LOO-CV, but it remains competitive with the standard cross-validation benchmark.

\begin{figure}[h!]
\centering
\includegraphics[width=0.72\textwidth]{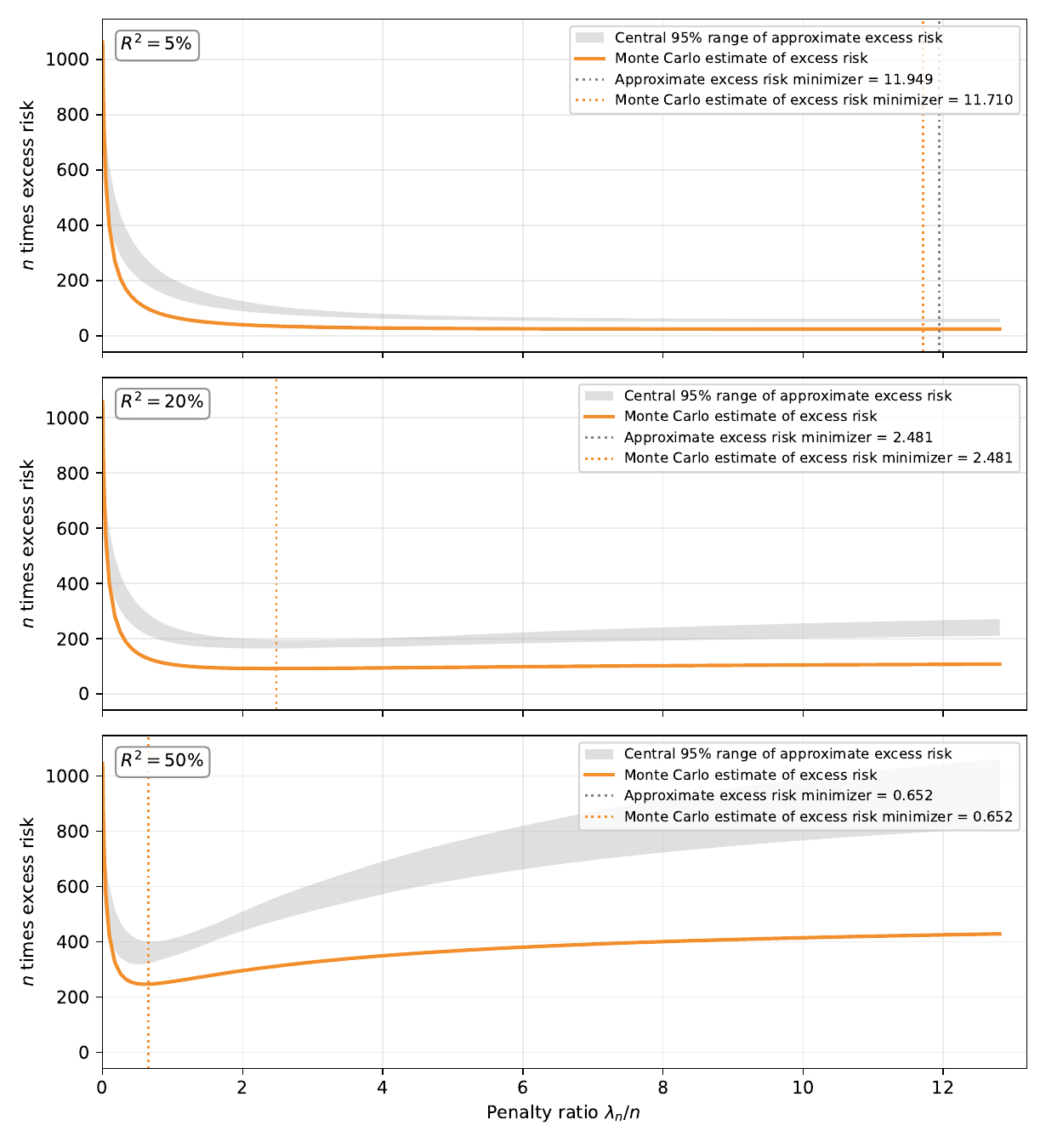}
\caption{Risk curves under DGP-3}
\label{fig:dgp3-risk-function}
\end{figure}

We next evaluate our approximation to the excess prediction risk of ridge estimators presented in Theorem~\ref{theorem: approximation to the prediction risk of ridge} under DGP-3. Figure~\ref{fig:dgp3-risk-function} displays a comparison between the Monte Carlo estimate of the scaled excess prediction risk, $n\left\{R(\widehat a_{\lambda_n};\beta_n,\mathbb P_n)-\sigma^2\right\}$ (solid orange line), and our scaled \emph{feasible} approximate excess risk $nR_n^e(\lambda;b,\widehat{\Sigma},\widehat{\Omega})$ (shaded region). To capture the range of the feasible approximation (across repetitions), we report the central 95 percent range after dropping the lowest 2.5 percent and highest 2.5 percent of values for each candidate tuning ratio. In this high-dimensional version of the design, the feasible approximation is less accurate than in DGP-2, especially away from the low-risk region. This pattern is expected because $k/n=0.6$ is far from the low-dimensional asymptotic framework used to derive our approximation. Appendix~\ref{subsec:dgp3-diagnostics} reports additional diagnostics for the eigenvalues of $\Sigma_2$, the coefficient vectors, and the selected tuning ratios.

\section{Conclusion} \label{sec:conclusion} 
We presented a simple Gaussian approximation to the finite-sample distribution of the ridge regression estimator. Our approximation is based on   nonstandard asymptotics where $i)$ we let the estimator's regularization parameter grow proportionally to the sample size; and $ii)$ we treat the population regression coefficients as \emph{local} to the reference vector that defines the estimator's direction of shrinkage. In contrast to other asymptotic approximations in the literature, we allow for general forms of heteroskedasticity and autocorrelation in the data generating process (at the cost of considering a low-dimensional model where the number of covariates is not allowed to grow with the sample size). We used our simple Gaussian approximation to propose two new strategies to select the regularization parameter for the ridge regression estimator. The suggested strategies select the regularization parameter to minimize either average or worst-case excess prediction risk, where risk is computed using our suggested Gaussian approximation.

\bibliographystyle{ecta}
\bibliography{references.bib}

\begin{thebibliography}{35}
\newcommand{\enquote}[1]{``#1''}
\expandafter\ifx\csname natexlab\endcsname\relax\def\natexlab#1{#1}\fi

\bibitem[\protect\citeauthoryear{Abadie and Kasy}{Abadie and Kasy}{2019}]{abadie2019choosing}
\textsc{Abadie, A. and M.~Kasy} (2019): \enquote{Choosing among Regularized Estimators in Empirical Economics: The Risk of Machine Learning,} \emph{The Review of Economics and Statistics}, 101, 743--762.

\bibitem[\protect\citeauthoryear{Andrews and Mikusheva}{Andrews and Mikusheva}{2022}]{andrews2022optimal}
\textsc{Andrews, I. and A.~Mikusheva} (2022): \enquote{Optimal decision rules for weak GMM,} \emph{Econometrica}, 90, 715--748.

\bibitem[\protect\citeauthoryear{Anup and Maddala}{Anup and Maddala}{1984}]{anup1984ridge}
\textsc{Anup, A.~K. and G.~Maddala} (1984): \enquote{Ridge estimators for distributed lag models,} \emph{Communications in Statistics-Theory and Methods}, 13, 217--225.

\bibitem[\protect\citeauthoryear{Atanasov, Zavatone-Veth, and Pehlevan}{Atanasov et~al.}{2024}]{atanasov2024risk}
\textsc{Atanasov, A., J.~A. Zavatone-Veth, and C.~Pehlevan} (2024): \enquote{Risk and cross validation in ridge regression with correlated samples,} \emph{arXiv preprint arXiv:2408.04607}.

\bibitem[\protect\citeauthoryear{Brockwell and Davis}{Brockwell and Davis}{2013}]{Brockwell_Davis:2013}
\textsc{Brockwell, P.~J. and R.~A. Davis} (2013): \emph{Time series: theory and methods}, Springer Science \& Business Media.

\bibitem[\protect\citeauthoryear{Cambanis, Huang, and Simons}{Cambanis et~al.}{1981}]{cambanis1981elliptical}
\textsc{Cambanis, S., S.~Huang, and G.~Simons} (1981): \enquote{On the theory of elliptically contoured distributions,} \emph{Journal of Multivariate Analysis}, 11, 368--385.

\bibitem[\protect\citeauthoryear{Cattaneo, Jansson, and Newey}{Cattaneo et~al.}{2018}]{cattaneo2018inference}
\textsc{Cattaneo, M.~D., M.~Jansson, and W.~K. Newey} (2018): \enquote{Inference in linear regression models with many covariates and heteroscedasticity,} \emph{Journal of the American Statistical Association}, 113, 1350--1361.

\bibitem[\protect\citeauthoryear{DasGupta}{DasGupta}{2008}]{Dasgupta08}
\textsc{DasGupta, A.} (2008): \emph{Asymptotic Theory of Statistics and Probability}, Springer Verlag.

\bibitem[\protect\citeauthoryear{Dobriban and Wager}{Dobriban and Wager}{2018}]{dobriban2018high}
\textsc{Dobriban, E. and S.~Wager} (2018): \enquote{High-dimensional asymptotics of prediction: Ridge regression and classification,} \emph{The Annals of Statistics}, 46, 247--279.

\bibitem[\protect\citeauthoryear{Dou and M{\"u}ller}{Dou and M{\"u}ller}{2021}]{dou2021generalized}
\textsc{Dou, L. and U.~K. M{\"u}ller} (2021): \enquote{Generalized Local-to-Unity Models,} \emph{Econometrica}, 89, 1825--1854.

\bibitem[\protect\citeauthoryear{Ferguson}{Ferguson}{1967}]{ferguson_1967}
\textsc{Ferguson, T.~S.} (1967): \emph{Mathematical Statistics: A Decision Theoretic Approach}, New York: Academic Press.

\bibitem[\protect\citeauthoryear{Friedman, Hastie, and Tibshirani}{Friedman et~al.}{2017}]{ESL2017}
\textsc{Friedman, J., T.~Hastie, and R.~Tibshirani} (2017): \emph{The elements of statistical learning: data mining, inference and prediction}, vol.~1 of \emph{Series in Statistics}, New York: Springer, second edition ed.

\bibitem[\protect\citeauthoryear{Gibbons}{Gibbons}{1981}]{gibbons1981}
\textsc{Gibbons, D.~G.} (1981): \enquote{A Simulation Study of Some Ridge Estimators,} \emph{Journal of the American Statistical Association}, 76, 131--139.

\bibitem[\protect\citeauthoryear{Hansen}{Hansen}{2022}]{hansen2022econometrics}
\textsc{Hansen, B.} (2022): \emph{Econometrics}, Princeton University Press.

\bibitem[\protect\citeauthoryear{Hansen}{Hansen}{2016}]{hansen2016efficient}
\textsc{Hansen, B.~E.} (2016): \enquote{Efficient shrinkage in parametric models,} \emph{Journal of Econometrics}, 190, 115--132.

\bibitem[\protect\citeauthoryear{Hastie, Montanari, Rosset, and Tibshirani}{Hastie et~al.}{2022}]{hastie2022surprises}
\textsc{Hastie, T., A.~Montanari, S.~Rosset, and R.~J. Tibshirani} (2022): \enquote{Surprises in high-dimensional ridgeless least squares interpolation,} \emph{The Annals of Statistics}, 50, 949--986.

\bibitem[\protect\citeauthoryear{Hirano and Wright}{Hirano and Wright}{2017}]{hirano2017forecasting}
\textsc{Hirano, K. and J.~H. Wright} (2017): \enquote{Forecasting with model uncertainty: Representations and risk reduction,} \emph{Econometrica}, 85, 617--643.

\bibitem[\protect\citeauthoryear{Hoerl}{Hoerl}{1962}]{horel1962application}
\textsc{Hoerl, A.~E.} (1962): \enquote{Application of ridge analysis to regression problems,} \emph{Chemical Engineering Progress}, 58, 54--59.

\bibitem[\protect\citeauthoryear{Hoerl and Kennard}{Hoerl and Kennard}{1970}]{hoerl1970ridge}
\textsc{Hoerl, A.~E. and R.~W. Kennard} (1970): \enquote{Ridge regression: Biased estimation for nonorthogonal problems,} \emph{Technometrics}, 12, 55--67.

\bibitem[\protect\citeauthoryear{Hoerl}{Hoerl}{2020}]{hoerl2020ridge}
\textsc{Hoerl, R.~W.} (2020): \enquote{Ridge regression: a historical context,} \emph{Technometrics}, 62, 420--425.

\bibitem[\protect\citeauthoryear{Hsu, Kakade, and Zhang}{Hsu et~al.}{2012}]{hsu2012random}
\textsc{Hsu, D., S.~M. Kakade, and T.~Zhang} (2012): \enquote{Random design analysis of ridge regression,} in \emph{Conference on learning theory}, JMLR Workshop and Conference Proceedings, 9--1.

\bibitem[\protect\citeauthoryear{Knight and Fu}{Knight and Fu}{2000}]{knight2000asymptotics}
\textsc{Knight, K. and W.~Fu} (2000): \enquote{Asymptotics for lasso-type estimators,} \emph{Annals of statistics}, 1356--1378.

\bibitem[\protect\citeauthoryear{Kock, Pedersen, and S{\o}rensen}{Kock et~al.}{2026}]{kock2026data}
\textsc{Kock, A.~B., R.~S. Pedersen, and J.~R.-V. S{\o}rensen} (2026): \enquote{Data-driven tuning parameter selection for high-dimensional vector autoregressions,} \emph{Journal of the American Statistical Association}, 121, 289--299.

\bibitem[\protect\citeauthoryear{Liu, Zheng, and Feng}{Liu et~al.}{2020}]{liu2020estimation}
\textsc{Liu, X., S.~Zheng, and X.~Feng} (2020): \enquote{Estimation of error variance via ridge regression,} \emph{Biometrika}, 107, 481--488.

\bibitem[\protect\citeauthoryear{Montiel~Olea, Rush, Velez, and Wiesel}{Montiel~Olea et~al.}{2026}]{montiel2026distributionally}
\textsc{Montiel~Olea, J.~L., C.~Rush, A.~Velez, and J.~Wiesel} (2026): \enquote{The distributionally robust prediction error of the LASSO and related estimators,} \emph{The Annals of Statistics}, 54, 1006--1027.

\bibitem[\protect\citeauthoryear{Mourtada}{Mourtada}{2022}]{mourtada2022exact}
\textsc{Mourtada, J.} (2022): \enquote{Exact minimax risk for linear least squares, and the lower tail of sample covariance matrices,} \emph{The Annals of Statistics}, 50, 2157--2178.

\bibitem[\protect\citeauthoryear{Mourtada and Rosasco}{Mourtada and Rosasco}{2022}]{mourtada2022elementary}
\textsc{Mourtada, J. and L.~Rosasco} (2022): \enquote{An elementary analysis of ridge regression with random design,} \emph{Comptes Rendus. Math{\'e}matique}, 360, 1055--1063.

\bibitem[\protect\citeauthoryear{Patil, Du, and Tibshirani}{Patil et~al.}{2024}]{patil2024optimal}
\textsc{Patil, P., J.-H. Du, and R.~J. Tibshirani} (2024): \enquote{Optimal ridge regularization for out-of-distribution prediction,} \emph{arXiv preprint arXiv:2404.01233}.

\bibitem[\protect\citeauthoryear{Patil, Wei, Rinaldo, and Tibshirani}{Patil et~al.}{2021}]{patil2021uniform}
\textsc{Patil, P., Y.~Wei, A.~Rinaldo, and R.~Tibshirani} (2021): \enquote{Uniform consistency of cross-validation estimators for high-dimensional ridge regression,} in \emph{International conference on artificial intelligence and statistics}, PMLR, 3178--3186.

\bibitem[\protect\citeauthoryear{Phillips}{Phillips}{1988}]{phillips1988regression}
\textsc{Phillips, P.~C.} (1988): \enquote{Regression theory for near-integrated time series,} \emph{Econometrica: Journal of the Econometric Society}, 1021--1043.

\bibitem[\protect\citeauthoryear{Powell}{Powell}{2017}]{powell2017identification}
\textsc{Powell, J.~L.} (2017): \enquote{Identification and Asymptotic Approximations: Three Examples of Progress in Econometric Theory,} \emph{Journal of Economic Perspectives}, 31, 107--124.

\bibitem[\protect\citeauthoryear{Shen and Xiu}{Shen and Xiu}{2025}]{shen2025weak}
\textsc{Shen, Z. and D.~Xiu} (2025): \enquote{Can Machines Learn Weak Signals?} Working Paper 33421, National Bureau of Economic Research, Cambridge, MA.

\bibitem[\protect\citeauthoryear{Staiger and Stock}{Staiger and Stock}{1997}]{staiger1997instrumental}
\textsc{Staiger, D. and J.~H. Stock} (1997): \enquote{Instrumental Variables Regression with Weak Instruments,} \emph{Econometrica}, 65, 557--586.

\bibitem[\protect\citeauthoryear{Swindel}{Swindel}{1976}]{swindel1976good}
\textsc{Swindel, B.~F.} (1976): \enquote{Good ridge estimators based on prior information,} \emph{Communications in Statistics-Theory and Methods}, 5, 1065--1075.

\bibitem[\protect\citeauthoryear{Velez}{Velez}{2024}]{velez2024asymptotic}
\textsc{Velez, A.} (2024): \enquote{On the asymptotic properties of debiased machine learning estimators,} \emph{arXiv preprint arXiv:2411.01864}.

\end{thebibliography}

\newpage

\appendix

\section{Proofs of Main Results} 

\subsection{Proof of Proposition \ref{proposition: small lambda_n approximation}} \label{subsec: small lambda_n approximation}

\begin{proof}
Let $X$ be the $n \times k$ matrix that contains $x_i^{\top}$ in its $i$-th row. Let $Y$ and $\epsilon$ be the $n \times 1$ vectors that contain $y_i$ and $\epsilon_i$ (respectively) in their $i$-th row. Define the matrices $\widehat{\Sigma} \equiv X^{\top}X/n$ and $\widehat{A} \equiv  \widehat{\Sigma} + (\lambda_n/n) \mathbb{I}_k$. Using this notation, $ \widehat{\beta}_{\lambda_n} = \widehat{A}^{-1} \left( X^\top Y/n + (\lambda_n/n) \beta_0\right)$ and $\widehat{\beta}_{OLS} = \widehat{\Sigma}^{-1} (X^\top Y)/n$. Algebra shows
\begin{align}
   \sqrt{n}\left( \widehat{\beta}_{\lambda_n} - \beta_n \right) &= {\widehat{A}}^{-1} \left(  X^\top \epsilon/\sqrt{n} -  (\lambda_n/n) \sqrt{n} \left(\beta_n -\beta_0 \right)  \right) \label{eqn: ridge sampling error}\\
     \sqrt{n}\left( \widehat{\beta}_{OLS} - \beta_n \right) &= \widehat{\Sigma}^{-1} X^\top  \epsilon/\sqrt{n}~.\notag
\end{align}
Therefore, 
$$\sqrt{n}\left( \widehat{\beta}_{\lambda_n} - \beta_n \right)-\sqrt{n}\left( \widehat{\beta}_{OLS} - \beta_n \right) = \left(\widehat{A}^{-1}  - \widehat{\Sigma}^{-1} \right)(X^\top  \epsilon)/\sqrt{n} - \widehat{A}^{-1}(\lambda_n/n) \sqrt{n} \left(\beta_n -\beta_0 \right)~,$$
which is $o_p(1)$ because (i) $\lambda_n/n \overset{p}{\rightarrow}0$ implies $\widehat{A}^{-1}  - \widehat{\Sigma}^{-1}$ is $o_p(1)$, (ii) Assumption \ref{asn: assumption high-level OLS ridge} implies $(X^\top  \epsilon)/\sqrt{n}$ is $O_p(1)$, and (iii) because we are assuming $(\lambda_n/n) \sqrt{n} \left(\beta_n -\beta_0 \right) \overset{p}{\to} 0$ in this proposition.
\end{proof}

\subsection{Proof of Theorem \ref{theorem: big lambda_n approximation}} 
\label{subsec: big lambda_n approximation} 

\begin{proof}
Recall that $\widehat{\Sigma} \equiv X^{\top}X/n$ and $\widehat{A} \equiv  \widehat{\Sigma} + (\lambda_n/n) \mathbb{I}_k$. Note that $\widehat{A}  \overset{p}{\rightarrow} \Sigma + \lambda \mathbb{I}_k $ since $\widehat{\Sigma} \overset{p}{\rightarrow} \Sigma$ (Assumption \ref{asn: assumption high-level OLS ridge})  and $\lambda_n/n \overset{p}{\rightarrow} \lambda$. Note also that $\widehat{A}^{-1} (\lambda_n/n) \sqrt{n} \left(\beta_n -\beta_0 \right) \overset{p}{\rightarrow} \left(\Sigma + \lambda \mathbb{I}_k\right)^{-1} \lambda b$. We conclude by using  Equation \eqref{eqn: ridge sampling error}, Assumption \ref{asn: assumption high-level OLS ridge}, and Slutsky's theorem. 
\end{proof}

\subsection{Proof of Theorem \ref{theorem: approximation to the prediction risk of ridge}} \label{subsec: approximation to Ridge prediction risk}
\begin{proof}
We establish the result in two steps. 
We first show that the  excess prediction risk $R_n(\hat{a}_{\lambda_n}; \beta_n, \mathbb{P}_n) $ is approximated by $R^e_n(\lambda;b, \Sigma, \Omega)$ up to a remainder error of size $o(n^{-1})$.  
We then show that $R^e_n(\lambda;b, \Sigma, \Omega) $, which depends on population parameters ($\Sigma$, $\Omega$), can be approximated by $R^e_n(\lambda;b, \widehat{\Sigma},\widehat{\Omega})$, which uses consistent estimators ($\widehat{\Sigma}$, $\widehat{\Omega}$), up to error of size $o_{(\beta_n,\mathbb{P}_n)}(1/n)$.  
The conclusion follows by combining these two steps and the triangle inequality. 

\emph{Step 1:} For a new draw $(x^{\top},\epsilon)$ from the stationary distribution $\mathbb{P}$ associated to $\mathbb{P}_n$, we use $\beta_n$ to construct a new outcome-covariate pair $(y,x)=(x^{\top}\beta_n+\epsilon,x )$. The exact finite-sample prediction risk of $\hat{a}_{\lambda_n}$ is:
\[R_n(\hat{a}_{\lambda_n}; \beta_n, \mathbb{P}_n) = \mathbb{E}_{(\beta_n, \mathbb{P}_n)}\big[(y-x^{\top}\widehat{\beta}_{\lambda_n})^2\big] = \sigma^2 + \mathbb{E}_{(\beta_n, \mathbb{P}_n)}\big[( \widehat{\beta}_{\lambda_n} - \beta_n)^{\top} \Sigma (\widehat{\beta}_{\lambda_n} - \beta_n)\big].\]
Defining $Z_n \equiv \sqrt{n}(\widehat{\beta}_{\lambda_n}-\beta_n)$, this becomes
\[R_n(\hat{a}_{\lambda_n}; \beta_n, \mathbb{P}_n) = \sigma^2 + \frac{1}{n}\mathbb{E}_{(\beta_n, \mathbb{P}_n)}\big[Z_n^{\top}\Sigma Z_n\big].\]
Note that \eqref{eqn: approximation to the prediction risk of ridge} can be rewritten as 
\[R^e_n(\lambda;b, \Sigma, \Omega) =  \frac{1}{n} \mathbb{E}_{Z}[Z^{\top}\Sigma Z],\]
where
\[ Z \sim \mathcal{N}_k(-(\Sigma+\lambda\mathbb{I}_k)^{-1}\lambda b,(\Sigma+\lambda\mathbb{I}_k)^{-1}\Omega(\Sigma+\lambda\mathbb{I}_k)^{-1}).\]
Therefore,
\[R_n(\hat{a}_{\lambda_n}; \beta_n, \mathbb{P}_n) - \sigma^2 -R^e_n(\lambda;b,\Sigma, \Omega) =\frac{1}{n} \Big(\mathbb{E}_{(\beta_n, \mathbb{P}_n)}\big[Z_n^{\top}\Sigma Z_n\big]-\mathbb{E}_{Z}[Z^{\top}\Sigma Z]\Big),\]
so it suffices to show
\[\mathbb{E}_{(\beta_n, \mathbb{P}_n)}\big[Z_n^{\top}\Sigma Z_n\big] \to \mathbb{E}_Z[Z^{\top}\Sigma Z].\]
Under the assumptions of Theorem \ref{theorem: approximation to the prediction risk of ridge}, Slutsky's theorem and Theorem \ref{theorem: big lambda_n approximation} imply that $Z_n \overset{d}{\to}  Z$.

Theorem 6.2 in \citet{Dasgupta08} implies that under assumption iii) of Theorem \ref{theorem: approximation to the prediction risk of ridge} (which is a sufficient condition for uniform integrability) we have $\mathbb{E}_{(\beta_n, \mathbb{P}_n)}\big[Z_n^{\top}\Sigma Z_n\big] \to \mathbb{E}_Z[Z^{\top}\Sigma Z]$.
This establishes the oracle approximation
\[R_n(\hat{a}_{\lambda_n}; \beta_n, \mathbb{P}_n) - \sigma^2 = R^e_n(\lambda;b,\Sigma, \Omega) + o\Big(\frac{1}{n}\Big).\]

\emph{Step 2:} It is sufficient to show that $n\left( R^e_n(\lambda;b, \widehat{\Sigma},\widehat{\Omega}) -  R^e_n(\lambda;b, \Sigma, \Omega)\right) = o_{(\beta_n,\mathbb{P}_n)}(1)$. We claim that this follows from the consistency of the estimators ($\widehat{\Sigma}$, $\widehat{\Omega}$) of ($\Sigma$, $\Omega$). To see this, note that
$$ n R^{e}_n( \lambda; b,\Sigma,\Omega ) \equiv  \lambda^2 b^{\top}  (\Sigma + \lambda \mathbb{I}_k)^{-1} \Sigma  (\Sigma + \lambda \mathbb{I}_k)^{-1} b +   \textrm{trace} \left(  (\Sigma + \lambda \mathbb{I}_k)^{-1} \Omega  (\Sigma + \lambda \mathbb{I}_k)^{-1} \Sigma \right).$$
Then, we conclude by using Slutsky's theorem and continuous mapping theorem. 
\end{proof}

\subsection{Proof of Theorem \ref{theorem: selection of lambda}} 
\label{subsec: selection of lambda} 

\begin{proof}
Under isotropic features, $\Sigma=\sigma_x^2 \mathbb{I}_k$, the approximation of the excess risk of ridge regression in \eqref{eqn: approximation to the prediction risk of ridge} simplifies to
\begin{equation}\label{eq:aux1-thm3}
    R^e_n(\lambda;b,\sigma_x^2,\Omega) = \frac{1}{n}\frac{\lambda^2\sigma_x^2}{(\sigma_x^2+\lambda)^2} b^\top b + \frac{1}{n}\frac{\sigma_x^2}{(\sigma_x^2+\lambda)^2}\operatorname{trace}(\Omega).
\end{equation} 
We use the previous closed-form expression to obtain
\begin{equation}\label{eq:aux2-thm3}
    \partial_\lambda R^e_n(\lambda;b,\sigma_x^2,\Omega)    
    = \frac{1}{n}\frac{2\sigma_x^2}{(\sigma_x^2+\lambda)^3}(\sigma_x^2 ( b^\top b) \lambda - \operatorname{trace}(\Omega))~.
\end{equation}
Note that the previous expression is negative if and only if $\lambda < \lambda^* \equiv \tfrac{\operatorname{trace}(\Omega)}{\sigma_x^2 ( b^\top b)}$. Therefore, $R^e_n(\lambda;b,\sigma_x^2,\Omega)$ as a function of $\lambda$ is decreasing on $[0,\lambda^*)$ and increasing on $[\lambda^*,\infty)$. As a result, we have that $\lambda^* \in \arg  \min_{\lambda \ge 0 } R^e_n(\lambda;b,\sigma_x^2,\Omega) $ for any given $(b,\sigma_x^2,\Omega)$.

We first establish part (1) of Theorem \ref{theorem: selection of lambda}. Note that when $\| \cdot \|$ is the Euclidean norm, the set $\mathcal{B} \equiv \{b \in \mathbb{R}^k : \|b\| \leq B\} = \{b \in \mathbb{R}^k : b^{\top}b \leq B^2\}$. Since the coefficient on $b^\top b$ in the excess risk approximation formula defined by \eqref{eq:aux1-thm3} is nonnegative for every $\lambda\ge 0$, the worst-case value of the approximate excess risk is attained when $b^\top b=B^2$. Hence the minimax problem reduces to minimizing the function
\[Q_{\textrm{minimax}}(\lambda) \equiv \frac{1}{n}\frac{\lambda^2\sigma_x^2}{(\sigma_x^2+\lambda)^2} B^2+\frac{1}{n}\frac{\sigma_x^2}{(\sigma_x^2+\lambda)^2}\operatorname{trace}(\Omega), \qquad \lambda\ge 0,\]
which as a function of $\lambda$ achieves its minimum at $  \lambda_{\mathrm{minimax}}^* = \tfrac{\operatorname{trace}(\Omega)}{\sigma_x^2 B^2}$.

We then establish part (2) of Theorem \ref{theorem: selection of lambda}. Suppose $\mathbb{E}_\pi[b^\top b]<\infty$. Taking expectation with respect to $\pi$ gives
\[Q_{\textrm{average}}(\lambda):=\mathbb{E}_\pi\bigl[R^e_n(\lambda;b,\sigma_x^2,\Omega)\bigr]=\frac{1}{n}\frac{\lambda^2\sigma_x^2}{(\sigma_x^2+\lambda)^2}\mathbb{E}_\pi[b^\top b]+\frac{1}{n}\frac{\sigma_x^2}{(\sigma_x^2+\lambda)^2}\operatorname{trace}(\Omega),\qquad \lambda\ge 0.\]
If $\mathbb{E}_\pi[b^\top b]>0$, algebra shows that $Q_{\textrm{average}}(\cdot)$ achieves its minimum at $  \lambda^*_{\textrm{average}} = \tfrac{\operatorname{trace}(\Omega)}{\sigma_x^2 \mathbb{E}_\pi[b^\top b]}$.
\end{proof}

\section{Additional Results} 

\subsection{Optimal Choice of $\lambda$ for Elliptically Contoured Distributions} \label{subsec:elliptical_distributions} 
In this appendix, we extend the result on choosing $\lambda$ to minimize average approximate excess risk to the more general case in which $b$ follows an elliptically contoured distribution. Following \citet{cambanis1981elliptical}, a $k$-dimensional random vector $X$ is said to follow an elliptical distribution with location parameter $\mu_X \in \mathbb{R}^k$, nonnegative definite matrix $V \in \mathbb{R}^{k \times k}$, and characteristic generator $\psi$ if the characteristic function of $X-\mu_X$ satisfies
\[
\phi_{X-\mu_X}(t)=\psi(t^\top V t), \qquad t \in \mathbb{R}^k.
\]
We denote this by $X \sim EC_k(\mu_X,V,\psi)$.

Now suppose $b \sim EC_k(0,\mathbb{I}_k,\psi)$ has finite second moments. Algebra shows that
\[
\mathbb{E}[bb^\top] = \frac{\mathbb{E}\|b\|^2}{k} \mathbb{I}_k.
\]
Therefore,
\begin{align*}
&\mathbb{E}_{b \sim \pi}\!\left[
b^\top (\widehat{\Sigma}+\lambda \mathbb{I}_k)^{-1}\widehat{\Sigma}(\widehat{\Sigma}+\lambda \mathbb{I}_k)^{-1} b
\right] \\
&\qquad = \mathbb{E}_{b \sim \pi}\!\left[ \operatorname{tr}\!\left( (\widehat{\Sigma}+\lambda \mathbb{I}_k)^{-1}\widehat{\Sigma}(\widehat{\Sigma}+\lambda \mathbb{I}_k)^{-1} bb^\top \right) \right] \\
&\qquad = \operatorname{tr}\!\left( (\widehat{\Sigma}+\lambda \mathbb{I}_k)^{-1}\widehat{\Sigma}(\widehat{\Sigma}+\lambda \mathbb{I}_k)^{-1} \mathbb{E}_{b \sim \pi}[bb^\top] \right) \\
&\qquad = \operatorname{tr}\!\left( (\widehat{\Sigma}+\lambda \mathbb{I}_k)^{-1}\widehat{\Sigma}(\widehat{\Sigma}+\lambda \mathbb{I}_k)^{-1} \frac{\mathbb{E}\|b\|^2}{k} \mathbb{I}_k \right) \\
&\qquad = \frac{\mathbb{E}\|b\|^2}{k} \operatorname{tr}\!\left((\widehat{\Sigma}+\lambda \mathbb{I}_k)^{-1}\widehat{\Sigma}(\widehat{\Sigma}+\lambda \mathbb{I}_k)^{-1} \right).
\end{align*}
Hence the average approximate excess risk takes the same form as in Section \ref{subsec:average}, with $C^2$ replaced by $\mathbb{E}\|b\|^2/k$. It follows that the corresponding characterization of the optimal choice of $\lambda$ extends immediately to the class of elliptical contoured  distributions with finite second moment.

\subsection{Additional Simulation Results}
\label{subsec: additional simulation results}

\subsubsection{DGP-1}

DGP-1 is the isotropic example. The main additional object of interest is the selected tuning ratio $\widehat\lambda_n/n$. Under the population version of the minimax rule in this design, the optimal scaled regularization parameter equals
\[
c^*=\frac{\operatorname{trace}(\Omega)}{\sigma_x^2B^2}=1.
\]

\begin{figure}[h!]
\centering
\includegraphics[width=0.72\textwidth]{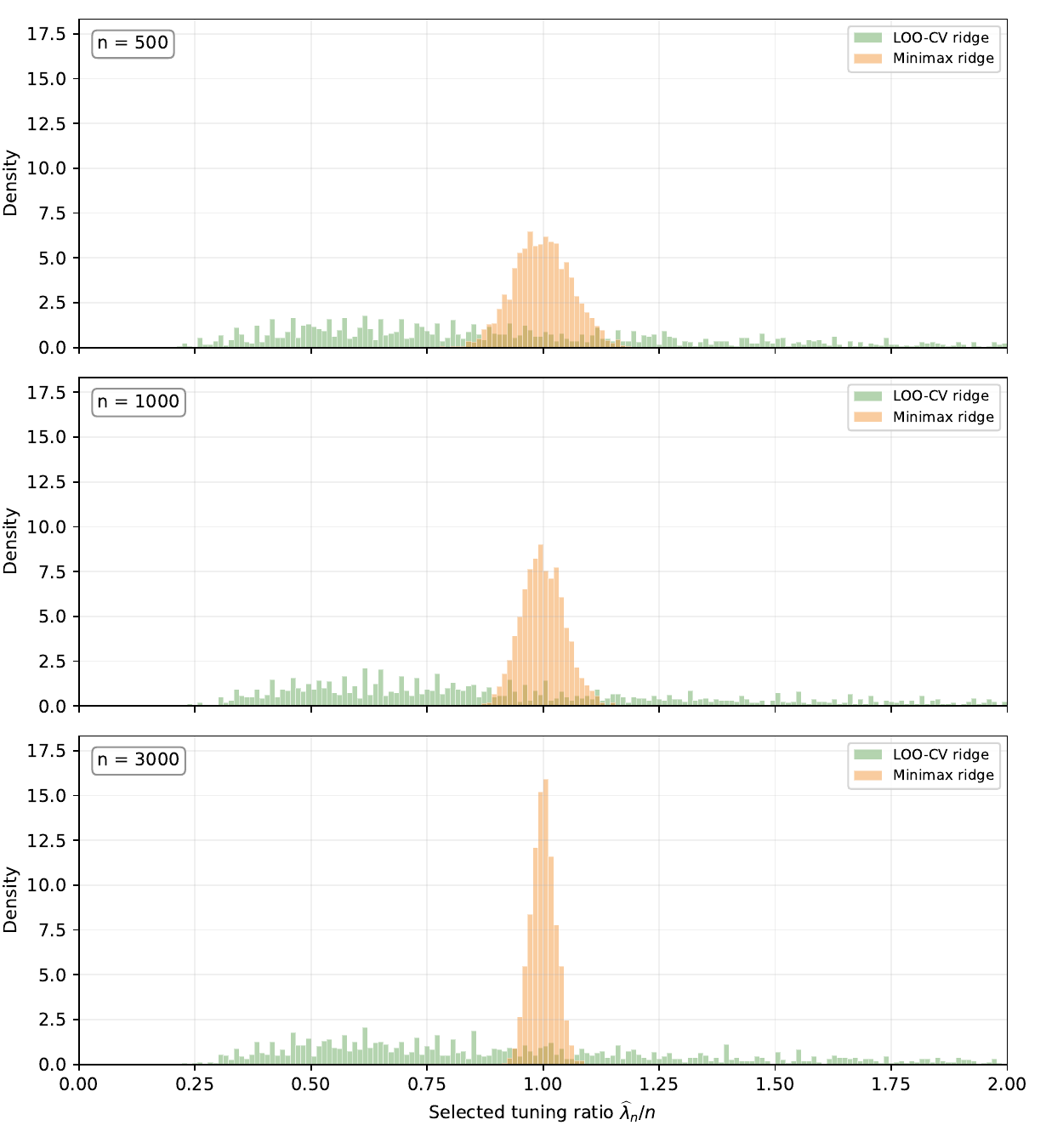}
\caption{Scaled regularization parameter ($\widehat\lambda_n/n$) used by the LOO-CV and minimax ridge estimators under DGP-1}
\label{fig:dgp1-ratio}
\end{figure}

Figure~\ref{fig:dgp1-ratio} reports the distribution of the selected tuning ratio $\widehat\lambda_n/n$ for LOO-CV ridge and minimax ridge. The feasible minimax rule is tightly centered near the population value $c^*=1$, with dispersion decreasing as the sample size increases. In contrast, the LOO-CV tuning ratio is much more dispersed across Monte Carlo repetitions. This helps explain the risk comparison: in DGP-1, minimax ridge uses a tuning parameter close to the risk-minimizing value, whereas LOO-CV often selects substantially different amounts of shrinkage.

\subsubsection{DGP-2}
\label{subsec:dgp2-dgp3-diagnostics}
\label{subsec:dgp2-diagnostics}

This appendix records diagnostic information for the covariance matrix and coefficient vectors used in DGP-2. These diagnostics are useful because the design departs from the isotropic benchmark: the prediction covariance is $\Sigma_2$, not a scalar multiple of the identity matrix, and the coefficient vector is a fixed draw from the spike-and-slab model described in the main text.

First, we verify the equality $\Sigma=\mathbb E[x_i x_i^\top]=\Sigma_2$ used in the simulation section. Let $e_i$ denote the $i$th canonical basis vector in $\mathbb R^n$. Since the $i$th row of $X$ is $x_i^\top=e_i^\top X$, we have
\begin{align*}
\mathbb E[x_i x_i^\top]
&=\mathbb E[X^\top e_i e_i^\top X] \\
&=\Sigma_2^{1/2}\,
\mathbb E\!
\left[Z^\top\Sigma_1^{1/2}e_i e_i^\top\Sigma_1^{1/2}Z\right]
\Sigma_2^{1/2}.
\end{align*}
For any deterministic $n\times n$ matrix $A$ and an $n\times k$ matrix $Z$ with i.i.d. standard normal entries, $\mathbb E[Z^\top A Z]=\operatorname{trace}(A)\mathbb{I}_k$. Applying this identity with $A=\Sigma_1^{1/2}e_i e_i^\top\Sigma_1^{1/2}$ gives
\begin{align*}
\mathbb E[x_i x_i^\top]
&=\Sigma_2^{1/2}\operatorname{trace}\!
\left(\Sigma_1^{1/2}e_i e_i^\top\Sigma_1^{1/2}\right)\mathbb{I}_k\Sigma_2^{1/2} \\
&=\Sigma_2^{1/2}(e_i^\top\Sigma_1 e_i)\mathbb{I}_k\Sigma_2^{1/2} \\
&=\Sigma_2^{1/2}(\Sigma_1)_{ii}\mathbb{I}_k\Sigma_2^{1/2} \\
&=\Sigma_2,
\end{align*}
where the last equality uses $(\Sigma_1)_{ii}=1$. Similarly, because the errors are independent of the regressors, mean zero, and have fixed variances $\sigma_{\epsilon,i}^2$, the corresponding variance matrix for the score is proportional to $\Sigma_2$:
\[
\frac{1}{n}\sum_{i=1}^n \mathbb E[x_i x_i^\top\epsilon_i^2]
=\left(\frac{1}{n}\sum_{i=1}^n \sigma_{\epsilon,i}^2\right)\Sigma_2.
\]

\begin{figure}[h!]
\centering
\includegraphics[width=0.72\textwidth]{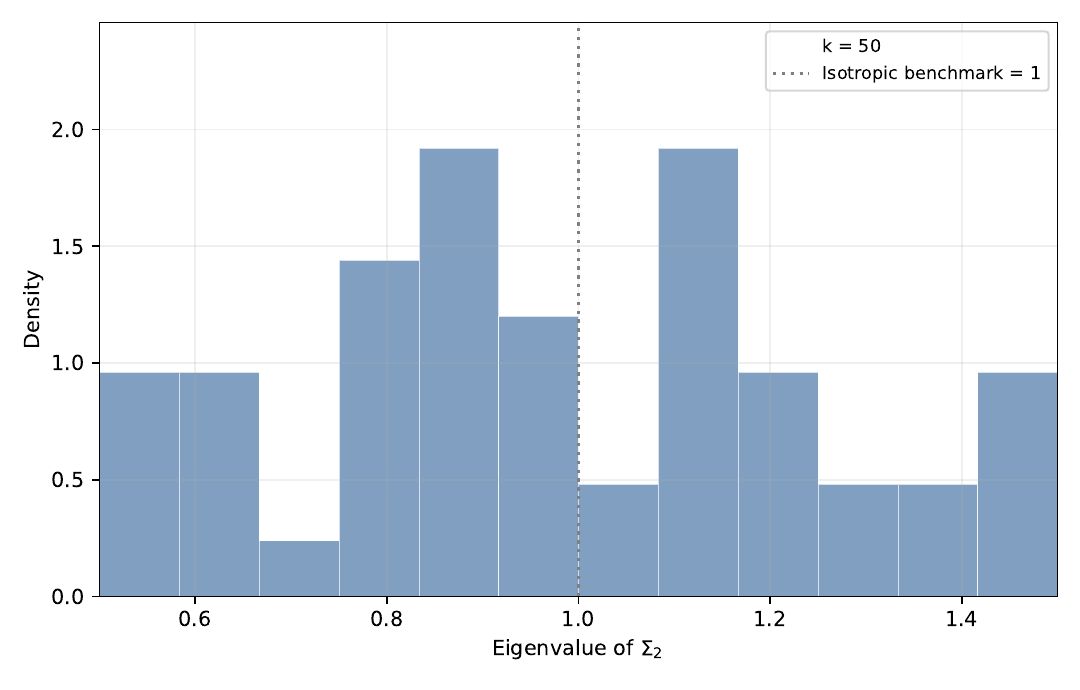}
\caption{Eigenvalues of $\Sigma_2$ under DGP-2}
\label{fig:dgp2-sigma2-eigenvalues}
\end{figure}

Figure~\ref{fig:dgp2-sigma2-eigenvalues} plots the eigenvalues of the fixed $\Sigma_2$ matrix used in DGP-2. The eigenvalues range from 0.502 to 1.479, with mean 0.981 and max-min spread 0.977. The vertical dashed line marks the isotropic benchmark value one. The figure shows that the design is centered near the isotropic benchmark on average but has meaningful dispersion in eigenvalues.

\begin{figure}[h!]
\centering
\includegraphics[width=0.72\textwidth]{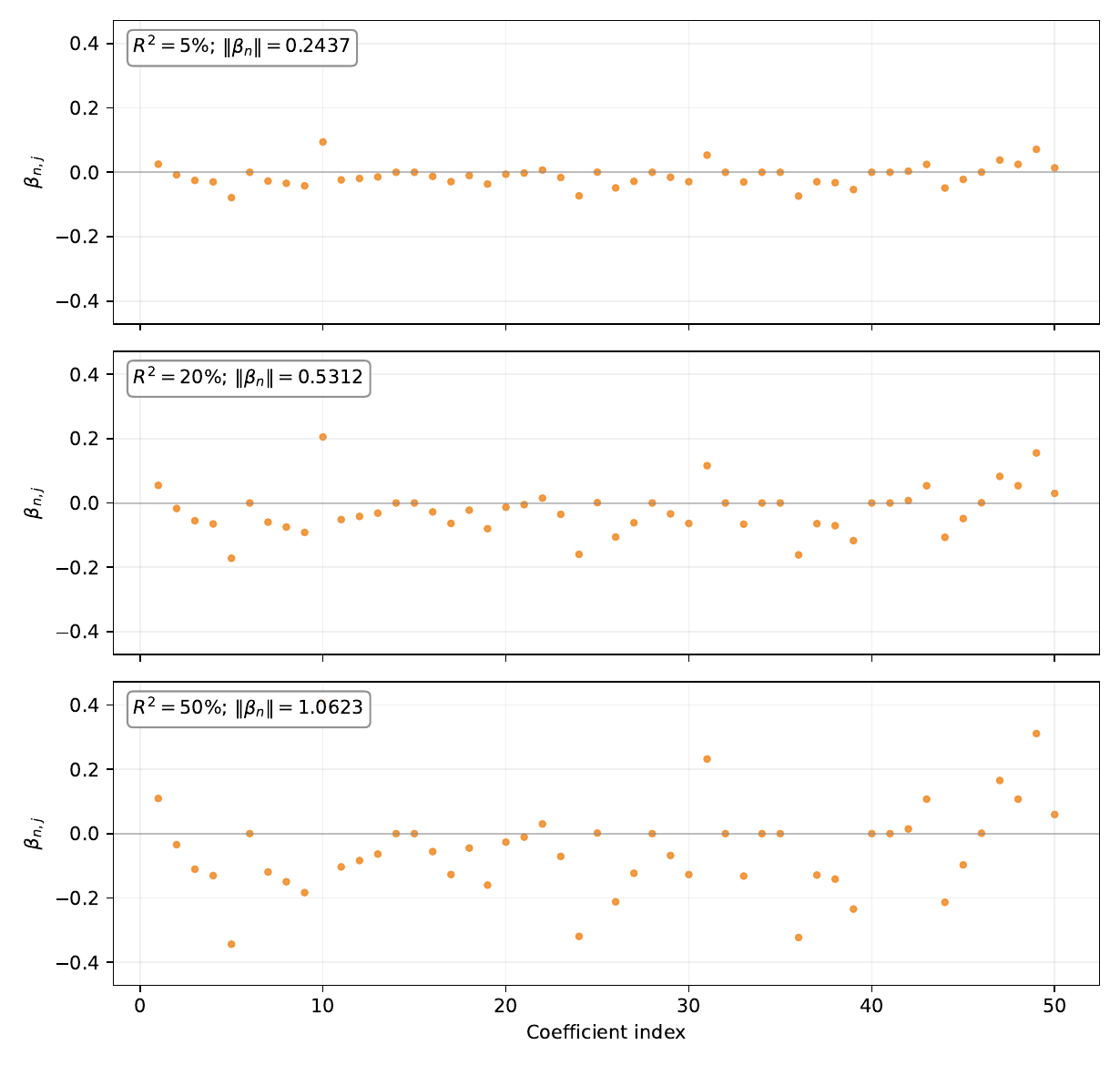}
\caption{Coefficient vectors under DGP-2}
\label{fig:dgp2-beta-coefficients}
\end{figure}

Figure~\ref{fig:dgp2-beta-coefficients} displays the fixed coefficient vector used in DGP-2 after rescaling the same preliminary draw to the three target values of $R^2$. The labels report $\|\beta_n\|$ for each target value. In this draw, 41 out of the 50 coordinates are nonzero. The figure makes explicit how increasing the target $R^2$ changes the magnitude of the coefficient vector while preserving the same coefficient pattern.

\begin{figure}[h!]
\centering
\includegraphics[width=0.72\textwidth]{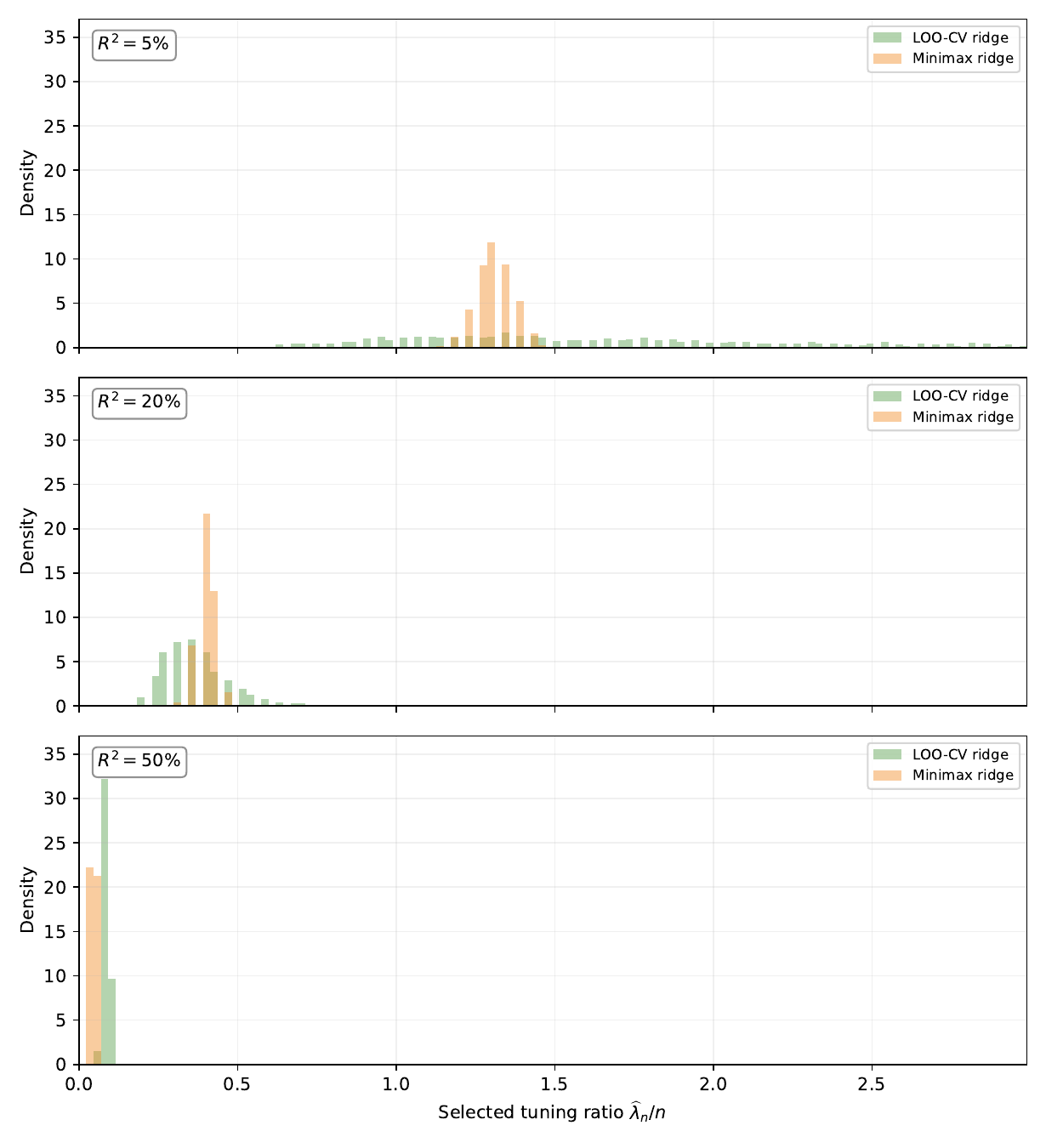}
\caption{Selected tuning ratios under DGP-2}
\label{fig:dgp2-selected-lambda}
\end{figure}

Figure~\ref{fig:dgp2-selected-lambda} reports the distribution of the selected tuning ratio $\widehat\lambda_n/n$ for LOO-CV ridge and minimax ridge under DGP-2. The minimax rule is more concentrated because it is driven by the approximate-risk calculation and the calibrated radius $B$, whereas LOO-CV is more dispersed across Monte Carlo repetitions.

\subsubsection{DGP-3}
\label{subsec:dgp3-diagnostics}

This appendix records diagnostic information for the covariance matrix and coefficient vectors used in DGP-3. DGP-3 uses the same covariance-generation and coefficient-generation schemes as DGP-2, but increases the number of covariates from $k=50$ to $k=300$ while keeping $n=500$. Thus, the diagnostics below should be read as the high-dimensional counterpart to Appendix~\ref{subsec:dgp2-diagnostics}.

The same algebra as in Appendix~\ref{subsec:dgp2-diagnostics} implies that the population prediction covariance is
\[
\Sigma=\mathbb E[x_i x_i^\top]=\Sigma_2,
\]
because the diagonal entries of $\Sigma_1$ are equal to one. Since the errors are independent of the regressors, mean zero, and have fixed variances $\sigma_{\epsilon,i}^2$, the corresponding score variance matrix is
\[
\frac{1}{n}\sum_{i=1}^n \mathbb E[x_i x_i^\top\epsilon_i^2]
=\left(\frac{1}{n}\sum_{i=1}^n \sigma_{\epsilon,i}^2\right)\Sigma_2.
\]

\begin{figure}[h!]
\centering
\includegraphics[width=0.72\textwidth]{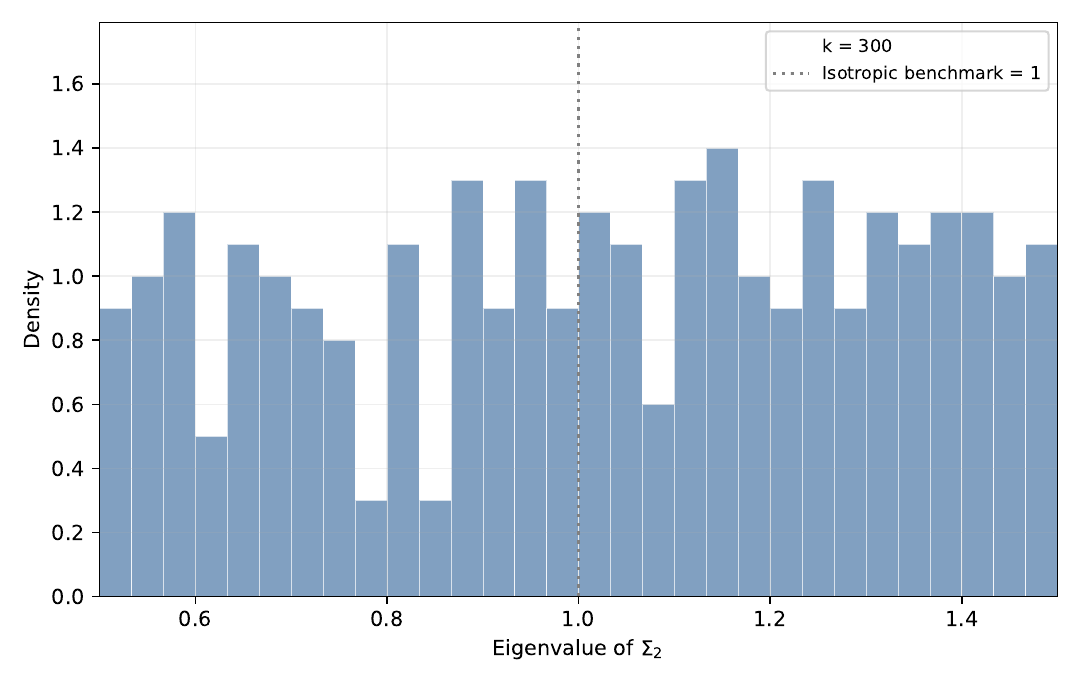}
\caption{Eigenvalues of $\Sigma_2$ under DGP-3}
\label{fig:dgp3-sigma2-eigenvalues}
\end{figure}

Figure~\ref{fig:dgp3-sigma2-eigenvalues} plots the eigenvalues of the fixed $\Sigma_2$ matrix used in DGP-3. The eigenvalues range from 0.502 to 1.499, with mean 1.027 and max-min spread 0.997. As in DGP-2, the design is centered near the isotropic benchmark on average but is clearly non-isotropic.

\begin{figure}[h!]
\centering
\includegraphics[width=0.72\textwidth]{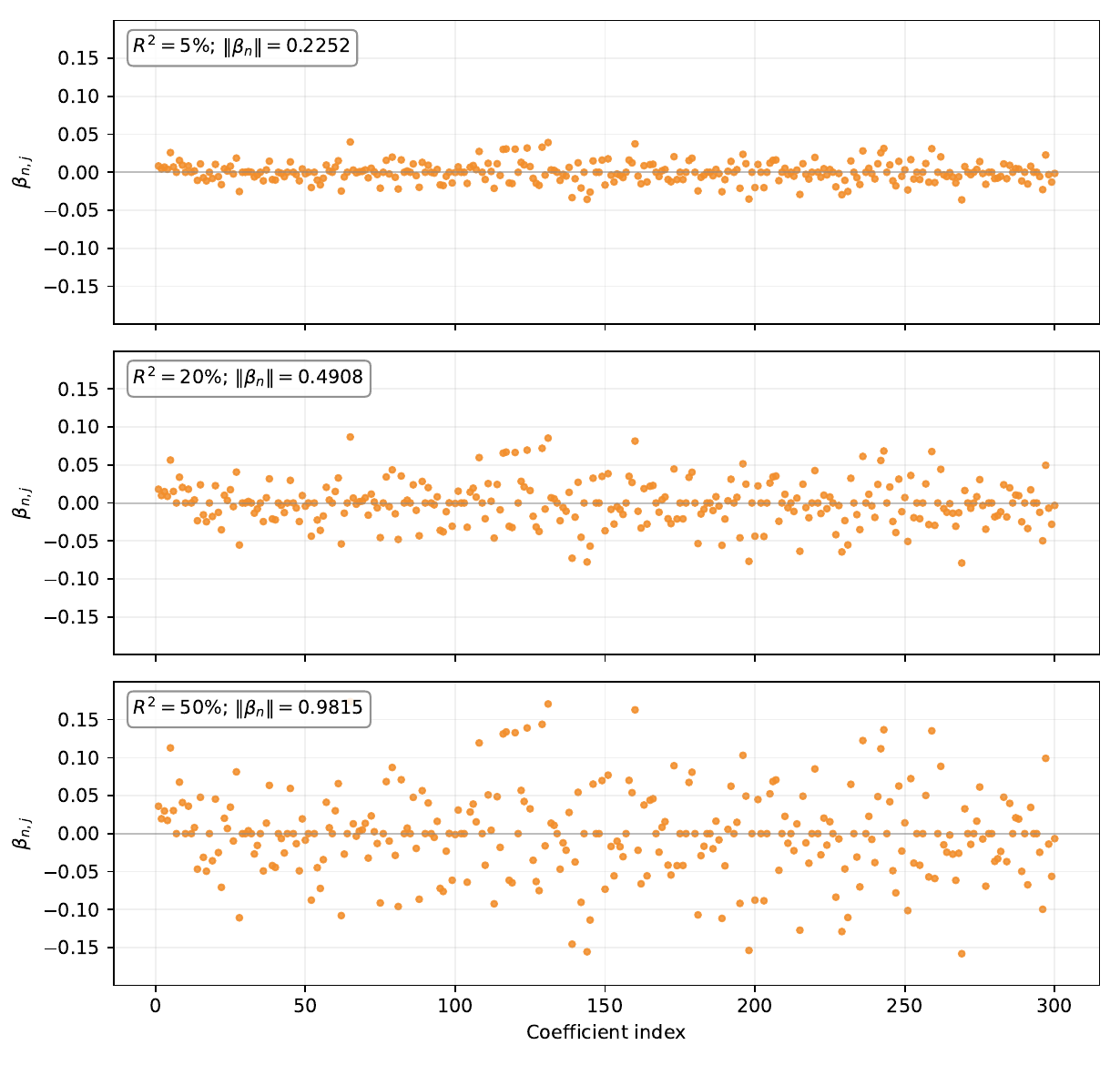}
\caption{Coefficient vectors under DGP-3}
\label{fig:dgp3-beta-coefficients}
\end{figure}

Figure~\ref{fig:dgp3-beta-coefficients} displays the fixed coefficient vector used in DGP-3 after rescaling the same preliminary draw to the three target values of $R^2$. The labels report $\|\beta_n\|$ for each target value. In this draw, 244 out of the 300 coordinates are nonzero. The figure shows that the signal is dense but individually weak, especially for the lower target value of $R^2$.

\begin{figure}[h!]
\centering
\includegraphics[width=0.72\textwidth]{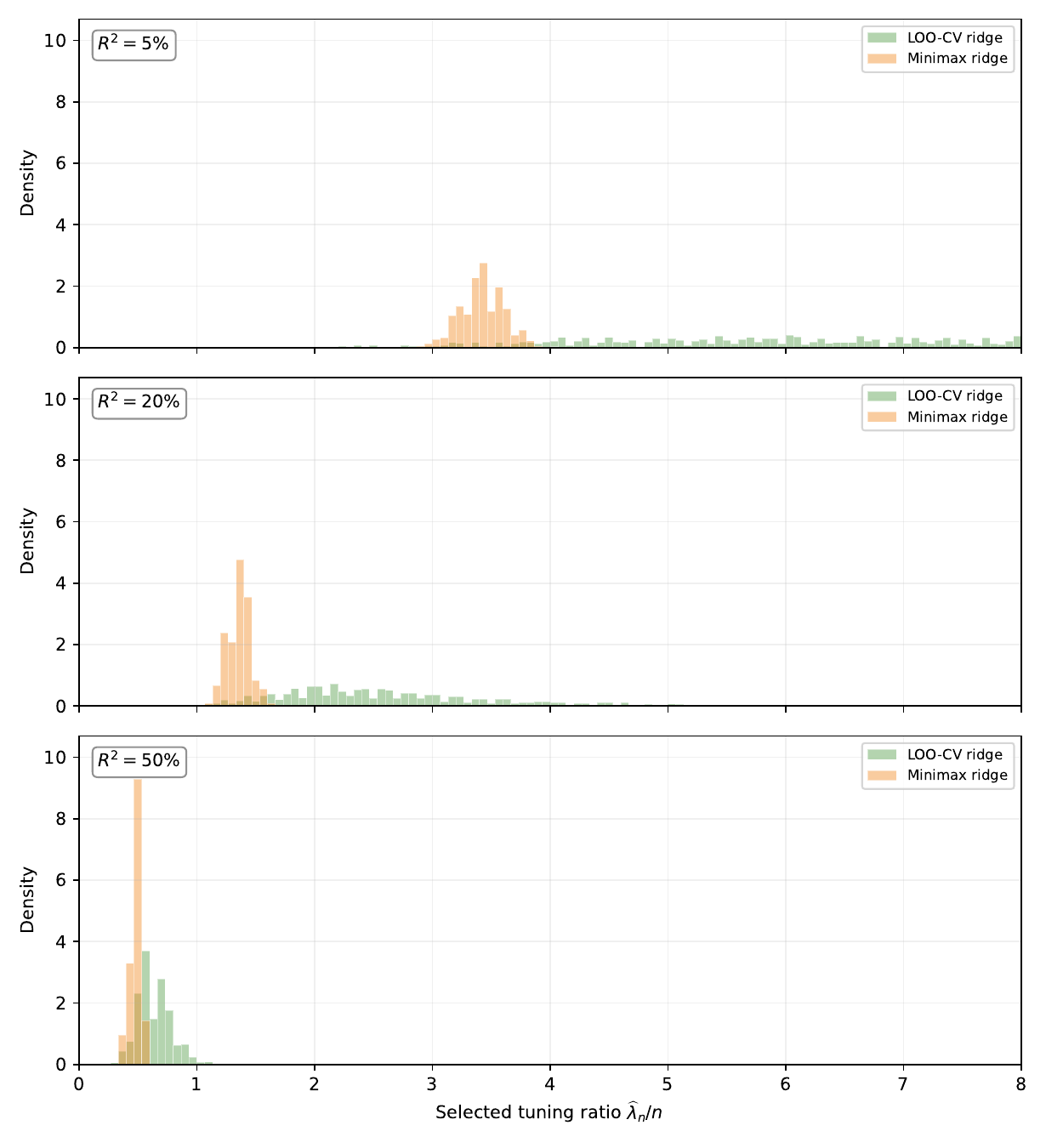}
\caption{Selected tuning ratios under DGP-3}
\label{fig:dgp3-selected-lambda}
\end{figure}

Figure~\ref{fig:dgp3-selected-lambda} reports the distribution of the selected tuning ratio $\widehat\lambda_n/n$ for LOO-CV ridge and minimax ridge under DGP-3. Relative to DGP-2, the LOO-CV choices are substantially more dispersed, especially when $R^2$ is small. This dispersion is consistent with the main-text finding that, in the high-dimensional design, LOO-CV and minimax ridge have comparable average risk but can choose very different penalty values in individual samples.

\newpage

\end{document}